\documentclass[review=false]{jfp-epi}

\usepackage{amsmath,stmaryrd}
\usepackage{amsthm}
\usepackage{enumerate}
\usepackage{bussproofs}
\usepackage{mathtools}
\usepackage{makecell}
\usepackage{comment}
\usepackage{proof}
\usepackage{graphicx}
\usepackage[misc]{ifsym}
\usepackage[shortlabels]{enumitem}
\usepackage{booktabs,dcolumn}
\usepackage{xspace}
\usepackage{colortbl}
\usepackage{hhline}
\usepackage{rotating}

\usepackage{xcolor}
\usepackage{float}

\usepackage{listings}
\usepackage{xcolor}
\def\withcolor{}
\def\incodesize{\small}
\newcommand{\codesize}{\small}

\ifdefined\withcolor
	\definecolor{haskellblue}{rgb}{0.0, 0.0, 1.0}
	\definecolor{haskellred}{rgb}{1.0, 0.0, 0.0}
	\definecolor{gray_ulisses}{gray}{0.55}
	\definecolor{castanho_ulisses}{rgb}{0.0,0.4,0.0}
	\definecolor{preto_ulisses}{rgb}{0.41,0.20,0.04}
	\definecolor{green_ulisses}{rgb}{0.8,0.0,0.8}
\else
	\definecolor{haskellblue}{gray}{0.1}
	\definecolor{haskellred}{gray}{0.1}
	\definecolor{gray_ulisses}{gray}{0.1}
	\definecolor{castanho_ulisses}{gray}{0.1}
	\definecolor{preto_ulisses}{gray}{0.1}
	\definecolor{green_ulisses}{gray}{0.1}
\fi

\lstdefinelanguage{HaskellUlisses} {
	basicstyle=\ttfamily\codesize,
	sensitive=true,
	mathescape=true,
	xleftmargin=\parindent,
	morecomment=[l][\color{gray_ulisses}\ttfamily\itshape\codesize]{--},
	morestring=[b]",
	literate={'"'}{\textquotesingle "\textquotesingle}3,
        literate={
           {<!}{{{\color{lcolor}<!}}}2
           {`}{{{$^{\backprime}{}$}}}1
           {?}{{{\color{lcolor}?}}}1
	   {\{-@\ }{{{}}}0
	   {\ @-\}}{{{}}}0
           {<=}{{$\leq$}}1
           {/=}{{$\neq$}}1
           {bot}{{$\bot$}}1
           {top}{{$\top$}}1
           {theta}{{$\theta$}}1
           {>>}{{{\color{haskellblue}>>}}}2
           {>>=}{{{\color{haskellblue}>>=}}}3
           {</>}{{{\color{haskellblue}</>}}}3
           {<*>}{{{\color{haskellblue}<*>}}}3
           {->}{{$\rightarrow$}}2
           {~>}{{$\imparrow$}}2
           {<-}{{$\leftarrow$}}2
           {<~}{{$\leftarrow$}}2
           {dollar}{{$\texttt{\$}$}}1
	},
	stringstyle=\color{haskellred},
	showstringspaces=false,
	numberstyle=\codesize,
	numberblanklines=true,
	showspaces=false,
	breaklines=true,
	showtabs=false,
	emph=
	{[1]
		FilePath,IOError,acos,acosh,all,any,appendFile,approxRational,asTypeOf,asin,
		sum,assert,checkSum1,checkSum2,sumTR',lemSumTR,thmSumTR,reachTrans,aval,bval,
		pre,sub,bsub,vc,ic,get,set,verify,
		lemSub,lemBsub,
		thmVC,lemVC,lemIC,lemIf,lemWhile,lemAsgn,thmFHLegit,
		plus, zeroL, zeroR, comm, succR,getEq, sorted, insert, merge, f, g, h
	},
	emphstyle={[1]\color{haskellblue}},
	emph=
	{[2]
		bool,char,int,nat,int8,list,olist,incList,decList,pair,incPair,
		BStep,Assertion,BExp,Com,Id,AExp,Val,State,Map,Bool,Char,Double,Either,
		Float,IO,Integer,Int,Maybe,Ordering,Rational,Ratio,ReadS,ShowS,String,
		Word8,Nat,NonZero,Nat64,Text,Map,ByteString,ByteStringSZ,ByteStringN,
    Ptr,ForeignPtr,CSize
    InPacket,Tree,Edge,Reach,ReachProp,ReachEv,Prop,TreeEq,TreeLt,Vec,
    NullTerm,IncrList,DecrList,UniqList,Map,BST,MinHeap,MaxHeap,
    PtrN,ByteStringN,ByteStringEq,VO,ByteStringsEq,ByteStringNE,
		List,Even,Valid,FH,Legit,Peano
	},
	emphstyle={[2]\color{castanho_ulisses}},
	emph=
	{[3]
		case,class,data,deriving,do,else,if,switch,return,import,in,infixl,infixr,instance,val,let,rec,
		requires,ensures,assume,val,module,measure,predicate,of,primitive,then,refinement,type,where,lazy
	},
	emphstyle={[3]\color{preto_ulisses}\textbf},
	emph=
	{[4]
		quot,rem,div,mod,elem,notElem,seq
	},
	emphstyle={[4]\color{castanho_ulisses}\textbf},
	emph=
	{[5]
		PS,Tip,Node,EQ,true,false,GT,Node,Leaf,Just,LT,Left,Nothing,Right,Show,Eq,Ord,Num,
		Z, S,Cons,Nil,OCons,ONil,Skip,Assign,Seq,While,If,BSkip,BAsgn,BSeq,BIfT,BIfF,BWhlT,BWhlF,
		Bc,Not,And,Leq,Eql,Plus,Minus,Times,V,N,FHSkip,FHAsgn,FHSeq,FHIf,FHWhl,FHPre,FHPost,
		Self,Step,Def,Set, True, False
	},
	emphstyle={[5]\color{green_ulisses}}
}

\lstnewenvironment{fscode}
{\lstset{language=HaskellUlisses,basicstyle=\ttfamily\small}}
{}

\lstnewenvironment{code}
{\lstset{language=HaskellUlisses}}
{}

\newcommand{\inlinespeccode}[1]{\lstinline[language=spec,mathescape=true]{#1}}

\lstnewenvironment{codebox}
{\lstset{language=HaskellUlisses,frame=tlbr,basicstyle=\ttfamily\codesize}}
{}

\lstMakeShortInline[language=HaskellUlisses,basicstyle=\ttfamily\incodesize]@

\lstdefinelanguage{spec}{
  basicstyle=\ttfamily\codesize,
  basewidth=0.5em,
  mathescape=true,
  xleftmargin=\parindent,
  morekeywords={bound,type,data,if,then,else,let,in,where,
                label,bind,return,unlabel,downgrade,toLabeled,
                project,proj1,selectList,select,
                respond,authUser
                },
  belowskip=5pt,
  aboveskip=5pt,
  mathescape
}

\lstnewenvironment{spec}{\lstset{language=HaskellUlisses}}{}

\lstnewenvironment{binahbox}{\lstset{language=spec,frame=single}}{}
\lstMakeShortInline[language=spec,mathescape=true]|

\usepackage{halloweenmath}

\usepackage[plain]{algorithm}
\usepackage[noend]{algpseudocode}
\definecolor{zgray}{RGB}{63,63,63}
\algloopdefx{Continue}[0]{\textbf{continue}}
\newcommand{\CodeComment}[2][\linewidth]{//~{\color{zgray}#2}}
\algblockdefx[SuchThat]{SuchThat}{EndSuchThat}{\textit{such that}}{}

\definecolor{darkcoral}{rgb}{0.8, 0.36, 0.27}
\definecolor{brown(traditional)}{rgb}{0.59, 0.29, 0.0}
\definecolor{darkmagenta}{rgb}{0.55, 0.0, 0.55}
\definecolor{nicepink}{rgb}{0.8,0.0,0.8}

\newcommand{\added}[1]{#1}
\newcommand{\rev}[1]{#1}

\newcommand{\todomp}[1]{}
\newcommand{\todoam}[1]{}

\newcommand{\commentout}[1]{}

\newif\iffull
\fulltrue 

\newcommand{\ite}{\mathit{ite}}
\newcommand{\nil}{\mathit{Nil}}
\newcommand{\cons}{\mathit{Cons}}

\newcommand{\halloc}[1]{\Fr^{x;v}}

\newcommand{\len}{\mathit{length}}

\newcommand{\height}{\mathit{height}}

\newtheorem{theorem}{Theorem}
\newtheorem{lemma}{Lemma}
\newtheorem{corollary}{Corollary}
\newtheorem{definition}{Definition}
\newtheorem{proposition}{Proposition}
\newtheorem{example}{Example}
\newsavebox{\qedtext}
\renewenvironment{proof}[1][]{%
    \sbox\qedtext{#1}%
    \par\addvspace{6pt}\noindent\textbf{Proof}\hskip5.5pt
}{\hfill\usebox{\qedtext}\hskip2pt$\blacksquare$\par\addvspace{6pt}}

\makeatletter
\newtheorem*{rep@theorem}{\rep@title}
\newcommand{\newreptheorem}[2]{%
\newenvironment{rep#1}[1]{%
 \def\rep@title{#2 \ref{##1}}%
 \begin{rep@theorem}}%
 {\end{rep@theorem}}}
\makeatother
\newreptheorem{theorem}{Theorem}

\makeatletter
\newtheorem*{rep@lemma}{\rep@title}
\newcommand{\newreplemma}[2]{%
\newenvironment{rep#1}[1]{%
 \begin{rep@lemma}}%
 {\end{rep@lemma}}}
\makeatother
\newreptheorem{lemma}{Lemma}

\makeatletter
\newtheorem*{rep@corollary}{\rep@title}
\newcommand{\newrepcorollary}[2]{%
\newenvironment{rep#1}[1]{%
 \begin{rep@corollary}}%
 {\end{rep@corollary}}}
\makeatother
\newreptheorem{corollary}{Corollary}

\newcommand{\listtype}{\mathsf{List}}
\newcommand{\inttype}{\mathsf{Int}}
\newcommand{\sorted}{\mathit{sorted}}
\newcommand{\reverse}{\mathit{rev}}
\newcommand{\extend}{\mathit{extend}}
\newcommand{\evenlist}{\mathit{evenlist}}
\newcommand{\ins}{\mathit{insert}}
\newcommand{\head}{\mathit{head}}
\newcommand{\tail}{\mathit{tail}}
\newcommand{\isnil}{\mathit{isNil}}
\newcommand{\iscons}{\mathit{isCons}}
\newcommand{\keys}{\mathit{elems}}
\newcommand{\mem}{\mathit{mem}}
\newcommand{\cont}{\mathit{Contract}}

\newcommand{\Ss}{\mathcal{S}}
\newcommand{\Ff}{\mathcal{F}}
\newcommand{\Dd}{\mathcal{D}}
\newcommand{\Tt}{\mathcal{T}}
\newcommand{\Mm}{\mathcal{M}}
\newcommand{\Nn}{\mathcal{N}}

\newcommand{\defn}{\mathit{def}\!}
\newcommand{\defns}{\mathit{DEF}}
\newcommand{\ctor}{\mathit{ctor}}
\newcommand{\dtor}{\mathit{dtor}}
\newcommand{\tup}{\overline{t}}
\newcommand{\thstd}{\Tt_\mathit{std}}
\newcommand{\thcomb}{\Tt_\mathit{comb}}

\newcommand{\smtproc}{\textsc{QFreeSAT}}

\newcommand{\qialg}{\textsc{UQFR}}

\newcommand{\formulas}{\mathit{formulas}}
\newcommand{\unfolded}{\Dd\_\mathit{applications}}
\newcommand{\infunfold}{\mathit{Inf}}

\newcommand{\liquid}{\textsc{Liquid Haskell}\xspace}
\newcommand{\leon}{\textsc{Leon}\xspace}
\newcommand{\stainless}{\textsc{Stainless}\xspace}

\newcommand{\sys}{\textit{FLUID}\xspace}
\newcommand{\lh}{\textsc{Liquid Haskell}\xspace}
\newcommand{\slh}{\textsc{LH}\xspace}
\newcommand{\mypara}[1]{\smallskip\noindent\emph{\textbf{#1}.\ }}

\newcommand{\fdef}[1]{\defn_{#1}}
\newcommand{\instat}[2]{{#1}[#2]}
\newcommand{\tPeano}{\mathit{Peano}}
\newcommand{\tplus}{\mathit{plus}}
\newcommand{\tpred}{\mathit{pred}}
\newcommand{\tget}{\mathit{get}}
\newcommand{\tset}{\mathit{set}}
\newcommand{\tJust}{\mathit{Just}}
\newcommand{\tZ}{\mathit{Z}}
\newcommand{\tS}{\mathit{S}}
\newcommand{\isZ}{\mathit{isZ}}
\newcommand{\isS}{\mathit{isS}}
\newcommand{\tIte}[3]{\mathit{ite}(#1,\ #2,\ #3)}
\newcommand{\ie}{i.e.\xspace}
\newcommand{\eg}{e.g.\xspace}
\newcommand{\vcl}{\mathit{Defs}}
\newcommand{\vcr}{\mathit{\varphi}}
\newcommand{\vc}{\vcl \rightarrow \vcr}
\newcommand{\ivc}{instantiated VC\xspace}
\newcommand{\univ}{\mathcal{U}}
\newcommand{\intp}[1]{\mathcal{I}(#1)}

\newcommand{\strat}{\mathit{strat}}
\newcommand{\pathcond}{\mathit{Path}}

\newcommand{\DApp}{\mathit{DApp}}

\jfpVolume{38}
\jfpArticle{3}
\jfpDOI{12345.67890}
\jfpYear{2026}
\jfpMonth{1}

\received[Submitted]{2007-02-20}
\received[accepted]{2009-06-05}

\begin{document}

\title{Complete first-order reasoning for functional programs}

\author{Adithya Murali}
\orcid{0000-0002-6311-1467}
\affiliation{%
  \institution{University of Wisconsin-Madison}
  \city{Madison}
  \country{USA}
  \authoremail{adithyamurali@cs.wisc.edu}
}

\author{Lucas Pe\~{n}a}
\orcid{0000-0002-1898-439X}
\affiliation{%
  \institution{University of Illinois Urbana-Champaign}
  \city{Urbana}
  \country{USA}
  \authoremail{lucaspena13@gmail.com}
}

\author{Ranjit Jhala}
\orcid{0000-0002-1802-9421}
\affiliation{%
  \institution{University of California San Diego}
  \city{San Diego}
  \country{USA}
  \authoremail{rjhala@eng.ucsd.edu}
}

\author{P. Madhusudan}
\orcid{0000-0002-9782-721X}
\affiliation{%
  \institution{University of Illinois Urbana-Champaign}
  \city{Urbana}
  \country{USA}
  \authoremail{madhu@illinois.edu}
}

\begin{abstract}
Several practical tools for automatically verifying functional programs (e.g., \liquid and \leon for Scala programs) rely on a heuristic based on unrolling recursive function definitions followed by quantifier-free reasoning using SMT solvers. We uncover foundational theoretical properties of this heuristic, revealing that it can be generalized and formalized as a technique that is in fact \emph{complete} for reasoning with combined First-Order theories of algebraic datatypes and background theories, where background theories support decidable quantifier-free reasoning. The theory developed in this paper explains the efficacy of these heuristics when they succeed, explains why they fail when they fail, and the precise role that user help plays in making proofs succeed.
\end{abstract}

\maketitle

\section{Introduction}
\label{sec:introduction}

The automation of program verification has been revolutionized
with the advent of efficient \emph{logic engines} that check
validity of logical formulas over various theories that capture
domains that programs work with (arithmetic, strings, arrays,
algebraic datatypes, pointer-based heaps, etc.).
In particular, \emph{quantifier-free logics} over various
theories admit decidable validity checking, and further, permit
decision procedures for the combination of theories
(Nelson-Oppen style combinations) that have been realized
by efficient DPLL(T)-based SMT solvers~\cite{calcofcomp,z3,nelson80,nelson-oppen1979}.

However, automation's grip becomes tenuous when it comes to the
verification of first-order properties of \emph{functional programs}
over \emph{algebraic data types} (ADTs) such as lists or trees
over basic types like integers.
Functional programs over ADTs can be expressed mathematically using
a set of \emph{recursively defined functions} over types.
Programs hence can be expressed as a set of first-order definitions of functions $\vcl$ that are \emph{universally quantified} over their inputs.
%
The goal of verification, then, is to determine whether a particular
FO (First-Order) theorem $T$ involving these defined (or interpreted) functions is
mathematically valid under a set of definitions $\vcl$.

\mypara{Automation is impossible in theory}
Even though the theorem $T$ that needs to be validated is universally
quantified (and hence can be seen as a quantifier-free formula),
reasoning about the validity of $T$ under \emph{interpreted definitions} $\vcl$
is extremely hard.
The validity problem is not decidable (while validity of $T$ under
\emph{uninterpreted functions} \emph{is} typically decidable).
Worse, the problem is not even recursively enumerable (there is
no complete proof system nor a semi-decision procedure that is
guaranteed to terminate on at least all valid theorems).
A simple proof of this fact is that we can define addition and
multiplication as defined (interpreted) functions using recursion,
and use universal quantification to specify the neither decidable nor
recursively enumerable problem of determining the \emph{non-existence} of
solutions to Diophantine equations~\cite{Matiyasevich}.

\mypara{Automation is effective in practice}
Despite the above hardness, there has been significant progress
in systems that provide varying degrees of automation to the
process of verifying such theorems.
\lh (\slh)~\cite{Vazou18} and \leon/\stainless~\cite{BlancKuncakKneussSuter,stainless} both exploit the automation provided
by logic engines via decidable \emph{quantifier-free reasoning}
to prove FO theorems.
Extrinsic-style verification in \slh reduces checking quantifier-free (implicitly universal) properties
of functions over ADTs to proving pre- and post-condition contracts
that assert those properties in the code (``proofs'') written by the
verification engineer.
\footnote{\slh implements various algorithms including refinement inference. 
In this work, when we say \slh, we refer specifically to extrinsic-style full functional correctness proofs over user-defined functions using methods proposed in \cite{Vazou18}.}
The \leon verifier \cite{BlancKuncakKneussSuter} (as well as its successor \stainless~\cite{stainless}) uses a similar style of reasoning for Scala programs
with quantifier-free contracts, where the contracts themselves are written using recursively defined pure Scala functions.
\leon verifies each function's contract by compiling the
body of the function to a \emph{verification condition} (VC),
modeling functions called in the body using defined functions
and assuming they satisfy their contracts
\footnote{This is an inductive proof (induction on the
size of the implicit call stack) that all functions satisfy their contracts.}.
While \slh and \leon provide different mechanisms for users to
prove properties via induction and auxiliary lemmas, we observe
that they share a common fundamental interface to logic engines:
verification is reduced to proving VCs of the
form $\vc$ where $\vcr$ is universally quantified.
\slh treats functions defined in $\vcl$ as \emph{uninterpreted}
using a heuristic called \emph{logical evaluation} that finitely
unfolds the definitions for terms that appear in $\vcr$.
\leon's strategy is also to unfold the recursive definitions based on
function applications that occur in $\varphi$. However, it differs
from \slh in that it does this \emph{recursively}, unfolding definitions iteratively for larger and larger depths and assuming that such
unfolded calls to functions satisfy their contract.

To summarize, both tools automate verification via logical
engines by
(1) generating VCs of the form $\vc$ (where $\vcr$ is a universally quantified FO formula),
(2) treating all defined functions as largely uninterpreted,
(3) instantiating definitions repeatedly only on certain terms, and
(4) dispatching them to an SMT solver that does quantifier-free decidable reasoning.
This technique is certainly sound but clearly not a decision procedure:
\slh just makes a fixed set of instantiations, which may be insufficient;
\leon can continuously unfold definitions and may proceed forever (timeout). 
Yet, despite the hardness results, this heuristic works well
in practice, giving predictable results though they may require
users to find new inductive lemmas and guidance in proofs!

\smallskip
%
\emph{Why does the heuristic of unfolding recursive
definitions followed by quantifier-free reasoning work so
well in practice?}
In this paper, we establish \emph{foundational results} that this procedure is in fact a \emph{complete}
procedure for the underlying combination of first-order theories.
Our results not only explain \emph{when} this heuristic method
works well, but also explains when and why they \emph{fail},
and the role of further help asked of the user.

\subsection*{The standard model vs combined theories}

The answer to this question lies in the tension between
the theory of the \emph{standard model} and the \emph{combined FO theory} of the various sorts. 
First-order theorems that express properties of functional
programs can be seen as formulas over a \emph{combination of sorts},
in particular sorts that refer to ADTs (\eg trees) and the base
sorts (\eg integers) that the datatypes are built upon.
When a verification engineer wishes to prove a theorem,
they want it to be proven for the \emph{fixed} universe (the standard model) consisting of the various sorts.
In this universe, the ADT sort is the natural universe
of \emph{algebraic terms} of the appropriate type, with
constructors and destructors interpreted in the standard
manner, and the integer sort and functions over them
(\eg $+$) are interpreted in the standard manner.

\mypara{Axiomatized models}
In first-order logic, however, we often reason with models
that are \emph{axiomatized}: we capture various
properties of models using a finite or recursive set of axioms,
and reason over \emph{any} model that satisfies the axioms.
In particular, ADTs can be axiomatized and the
universe of integers with addition can be axiomatized.
In fact, they can be individually given \emph{complete axiomatizations}--- \ie, all models satisfying the axioms satisfy the same first-order theorems as the standard model~\cite{barwise1982handbook,malcev62,hodges97,bjornerthesis,kovacs17,presburgertranslation2}. There may be other models, called \emph{nonstandard models}, that are not isomorphic to the standard model (in fact they always exist, say, by the L\"owenheim-Skolem Theorem) but one cannot distinguish them using a first-order formula. Nonstandard models are well-known in the literature~\cite{hodges97,skolem1934-nonstandard-models}.

\mypara{Rogue nonstandard models that disagree with the standard model}
However, when we combine universes and their theories, interfacing them with uninterpreted functions, the combined axioms are no longer powerful enough. 
More precisely, it is well known that the combined axioms can admit \emph{rogue} nonstandard models that disagree with standard models on first order expressible theorems, and hence the theory entailed by the combined axioms becomes weaker. Rogue nonstandard models are a special case of nonstandard models of the combined theory that disagree with the standard model on some first-order formulas. For example, if we take the complete axiomatization of ADTs and the complete axiomatization of uninterpreted functions (congruence axioms) and combine them, the union of the axioms admits rogue nonstandard models of ADTs that contradict theorems true in the standard model of ADTs with uninterpreted functions. \rev{For example, if we think of recursively defined functions over ADTs, the standard model would interpret them according to the least fixpoint of the definition, whereas nonstandard models could interpret them by arbitrary fixpoints\footnote{This is an only an intuitive explanation. As we will see in Section~\ref{sec:fluid}, we do not need any notion of least fixpoint to express definitions in the standard model. The high level idea is that since the standard model of ADTs consists of finite trees, well-behaved recursive functions will eventually recurse on smaller and smaller trees until they terminate, and First-Order Logic therefore suffices to capture their semantics accurately.}.}
%
%
%

Nonstandard models exist even for complete theories, but there are no rogue models in such theories since, by the definition of completeness, all nonstandard models agree with the standard model on all FO expressible theorems. Combined theories are incomplete as there are rogue nonstandard models. 
However, \emph{quantifier-free} formulas over combinations of theories (using the
Nelson-Oppen method) do not suffer from such issues,
which is why we can think of validity procedures for them as provers for the standard model.

\subsection*{Main contributions}

Our central insight in this paper is that
the method of unfolding recursive definitions
and performing quantifier-free reasoning can \emph{always} prove and \emph{only} prove the subset of theorems that
are valid over the \emph{combined theory of ADTs
and the background sorts}. Consequently, it \emph{cannot} prove theorems that are valid in the standard model but invalid in a rogue nonstandard model.
%
%
We develop this insight to explain the unusual
effectiveness of unfolding recursive definitions
into uninterpreted function applications,
via four concrete contributions.

\mypara{1. A \sys logic (\S~\ref{sec:fluid})}
Our first contribution is the definition of a logic called
\sys (First-order Logic for Universal properties under Inductive Definitions)
that captures the essence of definitions\footnote{\slh and \leon also support higher-order functions, but such definitions are beyond the scope of this paper. We provide further discussion on higher-order functions in Section~\ref{sec:discussion}.} and VCs
generated by \slh\footnote{In consultation with the developers of \slh, we believe that \sys captures all VCs generated by \slh for extrinsic style proofs using refinement reflection!} and \leon. 
\sys formulas are of the form $\vc$ where $\vcl$ are \emph{provably acyclic} recursive definitions, and $\varphi$ is a universally quantified formula. \rev{Intuitively, provably acyclic definitions are those for which one can prove termination in the usual sense using ranking functions, arguing that recursive calls decrease the rank (the subtlety being that we do not require ranks to be well-founded; see Section~\ref{sec:fluid}).} 
Verification conditions for correctness of many functional programs can thus be formulated using \sys formulas; in fact, systems like \slh and \leon generate VCs that are in this logical fragment. 

\mypara{2. Completeness of Unfolding followed by Quantifier-Free Reasoning (UQFR) (\S~\ref{sec:qialg})}
\qialg~ is a technique for proving validity by modeling recursive functions as uninterpreted functions, unfolding recursive definitions $\vcl$ systematically on a class of ground terms, and reasoning with the resulting quantifier-free formulae using decision procedures.
Our second contribution is a foundational result that shows
that \qialg~ is a
\emph{complete} semi-decision procedure for the validity of \sys formulas over the combined first-order theory of ADTs and background sorts. Namely, \qialg~ guarantees
to prove all theorems that are valid in the combined FO theory.
Consequently, when a theorem that is valid over the
standard model is \emph{not} proven using this technique,
we are guaranteed that there is a rogue nonstandard model
(satisfying the ADT and background theories)
where the theorem does not hold.
The proof of completeness is nontrivial
for two reasons.
First, the unfoldings of recursive definitions that
\qialg~ uses (and tools such as \slh and \leon use) are  \emph{thrifty}; they instantiate definitions of functions only on terms on which they are called, and do not expand instantiations to terms that arise from the underlying axiomatizations of theories. 
%
Second, every time a theorem that is valid on the
standard model is \emph{not} proven, it is nontrivial to 
construct a rogue nonstandard model falsifying the theorem.  The model construction in the proof of this theorem crucially exploits the fact that \sys definitions are provably acyclic. 
%

\mypara{3. Completeness in practice (\S~\ref{sec:lh})}
Thus, far from being a whimsical heuristic that happens
to work in practice, \qialg~ is rather a robust procedure
whose completeness may explain why this heuristic performs
so predictably well.
In particular, it does not miss proving theorems that can
be proved using pure FO reasoning of the underlying axioms
of the theories.
Our third contribution shows how this bears out in practice.
We explain how \slh performs \sys verification using \qialg.
Crucially, when theorems are not proved valid, we show it
is because rogue nonstandard models exist, and that the lemmas
and induction hints provided by the user then serve to
eliminate such models, all while reasoning within the \sys fragment.
%
Next, we show how we can use a slightly different \sys formula
to mimic \leon's more sophisticated reasoning which
additionally assumes pre/post contracts for functions
at each unfolding. Hence, our completeness
result also applies to explain the effectiveness of \leon (\S~\ref{sec:leon}).

\mypara{4. Limits of \sys (\S~\ref{sec:boundarythms})}
Our final contribution is a set of results that show why our results on \sys are unlikely to extend to more expressive logics.
We show though the validity problem for \sys admits complete procedures, it is \emph{undecidable}, hence distinguishing it from  several decidable fragments identified in the literature (\eg, \cite{suter10}).
We also show that attempts to generalize \sys, \eg by allowing functions whose definitions are required to be terminating (but not \emph{provably} so, using FO proofs) makes \qialg~ \emph{not} complete. This result also implies that replacing definitions with arbitrary universally quantified formulas makes \qialg~ an incomplete procedure.






\rev{
\smallskip
\mypara{Note on extended material} This manuscript extends a publication that appeared at OOPSLA 2023 under the title `Complete First-Order Reasoning for Properties of Functional Programs'. In addition to the material in the prior publication, this manuscript contains more detailed discussions on key technical concepts, formal constructions and intuitive descriptions of several rogue nonstandard models referenced in the prior publication as well as a few new rogue models, detailed proofs of all the theorems, and multiple explanatory discussions connecting the contributions of this paper to prior art. One significant addition was the formalization of the claims made in Section~\ref{sec:leon} and the addition of a `walkthrough' of a \leon-style proof of a VC, showcasing a rogue nonstandard model with a nontrivial construction that arises in this case. We thank the anonymous referees for their help in shaping this manuscript.
}
\section{Overview}
\label{sec:overview}

In this section we provide an overview of our work, which defines
\sys, a fragment of First-Order Logic (FOL) that expresses the quantified verification conditions
that arise when verifying correctness properties of functional programs
(see Section~\ref{sec:fluid}), and show how these VCs can be proved valid
by a \emph{thrifty} unfolding of definitions followed by quantifier-free reasoning.
Additionally, we show in \S~\ref{sec:lh} and \S~\ref{sec:leon} that verification
in systems like \slh and \leon reduces to checking validity of \sys fragment formulas.
We illustrate these ideas via example programs over the datatype of lists over integers:

\begin{code}
data List = Nil | Cons Int List
\end{code}




\subsection{Insertion and sortedness}
\label{sec:insert-example}

Consider the following program that inserts an element into
a (sorted) list 

\begin{code}
sorted :: List -> Bool
sorted Nil                    = True
sorted (Cons h Nil)           = True
sorted (Cons h1 (Cons h2 tl)) = h1 <= h2 && sorted (Cons h2 tl)

insert :: List -> Int -> List
insert Nil k                     = Cons k Nil
insert (Cons x xs) k | x >= k    = Cons k (Cons x xs)
                     | otherwise = Cons x (insert xs k)
\end{code}

\mypara{Definitions}
We can encode the above Haskell programs in FOL where each function's definition introduces no new variables, instead using destructors ($\head,\tail$) and recognizers ($\isnil,\iscons$) to simulate pattern matching. To ensure that destructors are applied sensibly, we \emph{guard} the use of terms of the form $\head(t)$ and $\tail(t)$ with the recognizer $\iscons(t)$.

$\forall x:\mathsf{List}.\, {\sorted}(x) = \;\ite(\isnil(x),\mathit{True},$

$\textrm{\hspace{8.5em}\;\;} \ite(\isnil(\tail(x)),\mathit{True},$

$\textrm{\hspace{15.5em}\;\;}\head(x) \leq \head(\tail(x)) \land \sorted(\tail(x))))$

$\forall x: \mathsf{List},\, k: \mathsf{Int}.\, \ins(x, k) = \;\ite(\isnil(x), \cons(k, \nil),$

$\textrm{\hspace{12.25em}\;\;\,}\ite(\head(x) \geq k, \cons(k, x), \cons(\head(x), \ins(\tail(x),k))))$
%

\noindent
where we treat $\sorted$ and $\ins$ as uninterpreted functions in the signature. 
We refer to these formulae as the \emph{definitions} of $\sorted$ and $\ins$ and denote them by $\defn_\sorted$ and $\defn_\ins$ respectively.

\mypara{Verification conditions}
Let us consider the example of verifying that inserting an element $k$ into the empty list yields a sorted list.
We state this formally as the following \emph{verification condition} (VC):

\centerline{
$\left(\defn_\sorted \land \defn_\ins\right) \rightarrow \sorted(\ins(\nil,k))$
}
Note that this VC is of the form $\defns \rightarrow \varphi$, where $\defns$ is a set of definitions (when it appears in a formula we are referring to the conjunction of the formulas in the set) and $\varphi$ is quantifier-free, \ie, all variables are
implicitly universally quantified. Informally, the VC says that the property $\varphi$ should
hold \emph{assuming} the set of definitions $\defns$. The \sys fragment we define (see Section~\ref{sec:fluid}) consists of such formulas.

\mypara{Unfolding}
We prove the above VC valid by \emph{unfolding} the definitions. 
For a term $t$, let $\defn_{\sorted}[t]$ denote the quantifier-free formula
obtained by instantiating the quantified variable $x$ in $\defn_{\sorted}$ with $t$. We refer to this as unfolding the definition of $\sorted$ on $t$. Similarly we can define the unfolding $\ins[\tup]$.
To prove the VC valid we simply unfold definitions on arguments that occur in $\varphi$,
\ie, we attempt to prove
\begin{center}
    $\left(\defn_{\ins}[(\nil,k)] \land \defn_{\sorted}[\ins(\nil,k)]\right) \rightarrow \sorted(\ins(\nil,k))$
\end{center}
This formula can be dispatched using SMT solvers~\cite{z3,cvc4}
that use a combination of decision procedures for ADTs and Integers. It is in fact valid because unfolding $\defn_\ins$ on $(\nil,k)$ shows that $\ins(\nil,k)$ evaluates to $\cons(k,\nil)$, and unfolding $\defn_\sorted$ on $\cons(k,\nil)$ shows that $\sorted(\ins(\nil,k))$ evaluates to $\mathit{True}$.

We generalize this technique of Unfolding definitions followed by Quantifier-Free Reasoning
into an algorithm \qialg~ (Section~\ref{sec:qialg}), and argue that tools like
\liquid~ (Section~\ref{sec:lh}) and \leon (Section~\ref{sec:leon}) perform similar reasoning on such formulas.

\subsection{Insertion preserves sortedness}
\label{sec:insert-example-contracts}

Next, let us turn to a more interesting theorem, namely that 
insertion preserves sortedness.
Formally, we wish to prove the following \emph{contract} for insertion:
$\forall x, k. \sorted(x) \rightarrow \sorted(\ins(x, k))$. 
%
Here the VC is 
$\mathit{VC}_{\mathit{simple}} \equiv \left(\defn_\sorted \land \defn_\ins\right) \rightarrow (\sorted(x) \rightarrow \sorted(insert(x,k)))$

Unlike the example in Section~\ref{sec:insert-example}, it turns out that there is no set of terms such that unfolding the definitions on these terms can prove the VC valid. Consequently, \slh fails to prove the theorem.

\mypara{Using contracts} Tools like \leon not only unfold definitions but also use contracts for terms generated during unfolding. For example, note that unfolding $\defn_\ins$ on $(x,k)$ yields the term $\ins(\tail(x),k)$. 
Then, the VC that \leon attempts to prove is not $\mathit{VC}_{\mathit{simple}}$ but rather  
\begin{align*}
    \mathit{VC}_{\leon} \equiv \left(\defn_\sorted \land \defn_\ins\right) \rightarrow \big(\left(x \neq \nil \right. & \left. \rightarrow \left(\sorted(\tail(x))\rightarrow \sorted(\ins(\tail(x),k))\right)\right)\\
    &\rightarrow \left(\sorted(x) \rightarrow \sorted(\ins(x,k))\right)\big)
\end{align*}

\noindent
which additionally assumes the contract for $\ins(\tail(x),k)$ (when $x \neq \nil$). Observe that this VC is also of the form $\defns \rightarrow \varphi$ and is therefore in the \sys fragment. We show in Section~\ref{sec:leon} that $\mathit{VC}_{\leon}$ can be obtained automatically from the original VC, \ie, $\mathit{VC}_{\mathit{simple}}$. 

We attempt to prove $\mathit{VC}_{\leon}$ using the same technique of unfolding recursive definitions on arguments appearing in the formula (\qialg). This succeeds, and one can verify that unfolding $\ins$ on $\{(x,k), (\tail(x),k)\}$ and unfolding $\sorted$ on $\{x,\tail(x),\ins(x,k), \ins(\tail(x),k)\}$ proves $\mathit{VC}_{\leon}$ valid\footnote{The reader may note here that we only argue the validity of $\mathit{VC}_{\leon}$ and not the original goal $\mathit{VC}_{\mathit{simple}}$. We discuss why validity of the former implies validity of the latter in Section~\ref{sec:leon}.}. Observe that the unfolding strategy used in the examples we have seen is \emph{thrifty} in the sense that definitions are unfolded \emph{exactly} on terms that occur as arguments to the corresponding functions. We discuss the utility of this strategy in Section~\ref{sec:discussion}.

It is clear that using contracts is a more powerful approach. In general, there are theorems whose proofs require even more instantiations of contracts on terms obtained during further unfoldings. We show in Section~\ref{sec:leon} using a reduction that the use of multiple repeated instantiations of definitions as well as contracts can also be viewed as proving \sys fragment formulas using \qialg. Consequently, our results apply not only to \slh but also to tools like \leon.

\subsection{Membership in a sorted list}
\label{sec:sorted-list-membership}

One prominent aspect of program verification in \slh or \leon is proof by induction. However, induction is \emph{not} part of \qialg. In this section we discuss the example of checking membership in a sorted list where all the above approaches fail, and explain the role of induction (in the form of explicit user help) in these tools. We first define a function $\keys$ to capture the set of elements stored in a list and a function $\mem$ that checks the membership of an element in a sorted list.



    $\forall x:\listtype.\, \keys(x) = \ite(\isnil(x), \emptyset, \{\head(x)\} \cup \keys(\tail(x)))$
    
$\forall x:\listtype,\,k:\inttype.\, \mem(x,k) = \ite(\isnil(x), \mathit{False},\,\ite(k = \head(x),\mathit{True},$

$\textrm{\hspace{19em}\;\,\,\,}\ite(k < \head(x), \mathit{False},$

$\textrm{\hspace{26em}\;\;}\mem(\tail(x),k))))$



We want to verify that $\mem$ precisely captures membership for sorted lists. Formally, the contract is:\; $\sorted(x) \rightarrow (\mem(x,k) \leftrightarrow k \in \keys(x))$

However, the approaches discussed above do not work for this example. They do not succeed even if the definitions are unfolded infinitely and contracts are assumed for all of the infinitely many terms/tuples that occur in the unfoldings.

To see why this is the case, consider what happens when we replace the usual \emph{standard} model of ADTs and Integers we have in our minds with complete axiomatizations for each of the sorts, along with congruence axioms for the function symbols $\keys$ and $\mem$. In this setting, the standard model is only one of the possible models and in general a model of the axiomatized universes may not be identical to the standard model. 
Validity in the axiomatized setting is an under-approximation to validity in the standard model, and as we show in Section~\ref{sec:combined-theories} it is in fact a strict under-approximation. There are theorems that are true on the standard model that do not hold under axiomatization. This is because of the presence of \emph{rogue nonstandard models} 
where the property we want to prove is not true. A rogue nonstandard model is a model that obeys the axioms but is not identical to the standard model, and further, falsifies the desired theorem. Nonstandard models always exist in the axiomatized setting, but they may satisfy all the same first-order properties as the standard model using a first-order formula. However, \emph{rogue} nonstandard models, when they exist, can disagree with the standard model on a desired first-order theorem.

Our soundness and completeness results in Section~\ref{sec:qialg} show that the proving power of unfolding definitions and using contracts is \emph{precisely} that of validity over the axiomatized universe. Therefore, if there is a rogue nonstandard model that falsifies a property, then unfolding based reasoning \emph{cannot} prove it. Indeed, both \slh and \leon fail on the above example without extra help.

\mypara{Rogue nonstandard model} Let us look at the rogue nonstandard model where our theorem does not hold. 
The universe $U$ is:
\begin{align*}
&\{s \mid s \textrm{ is a finite sequence of integers}\}\\
\cup\,& \{(s,i) \mid s \textrm{ is an infinite sequence of integers}, i \textrm{ is an integer}\}   
\end{align*}
%
%
The finite sequences correspond to ADT lists of integers as we think of them but the infinite sequences are \emph{nonstandard elements}\footnote{The standard model of ADTs consists exactly of all terms. \emph{Nonstandard elements} are elements in a nonstandard model that do not correspond to any term. In particular, one cannot destruct them a finite number of times to reach $\nil$. Nonstandard models always have such elements with ``infinite tails''.}. $\nil$ is interpreted to be the empty sequence and $\cons$ behaves as expected on standard elements (prepending an element to a finite sequence). On the nonstandard elements $\cons$ is defined by $\cons(j,(s,i)) = (j::s,i+1)$ where $j::s$ denotes prepending $j$ to the sequence $s$. $\head$ and $\tail$ behave as inverses to $\cons$ in the usual sense. One can check easily that this model satisfies the usual axioms of ADTs~\cite{hodges97,bjornerthesis}.

The meaning of $\sorted$ on this model is as expected: we define only elements with non-decreasing sequences to be sorted. 
The definition of $\keys(x)$ is as follows
\begin{center}
$\keys(x) = \begin{cases}
  \{ v \,\mid\, \text{$v$ is an element of $x$} \} \quad &\text{for a standard element } x \\
  \{ v \,\mid\, \text{$v$ is an element of $x$} \} \cup \{-1\} \quad &\text{for a nonstandard element } (x,i)
\end{cases}$
\end{center}

Lastly $\mem(x, k)$ holds if and only if $k$ occurs in the longest non-decreasing prefix of the sequence corresponding to $x$. Note that if $x$ is sorted, the longest non-decreasing prefix of $x$ is $x$ itself.


The above interpretations are consistent with the definitions. Consider the function for $\keys$, for example. On standard elements it is consistent with the definition because it has the expected value. It is also consistent on nonstandard elements. Observe that for a nonstandard element $x$, $\tail(x)$ is also a nonstandard element. Therefore, the inclusion of an extraneous element $-1$ in the $\keys$ of both $x$ and $\tail(x)$ is consistent with the recursion $\keys(x) = \{\head(x)\} \cup \keys(\tail(x))$.

Finally, we see this is a rogue nonstandard model as it does not satisfy the property $\sorted(x) \rightarrow (\mem(x,k) \leftrightarrow k \in \keys(x))$. 
Consider the nonstandard element $x = ([0,0,0\ldots],0)$. Note that $x$ is sorted since it is a non-decreasing sequence, and $\keys(x) = \{0,-1\}$ by the above construction. Hence $-1 \in \keys(x)$. However $\mem(x,-1) = \mathit{False}$ since $-1$ does not occur in $x$.


\mypara{Role of user help} To prove the above example in \slh, one must provide additional hints or \emph{inductive lemmas} (whose proof of the induction step is itself performed using unfolding/\qialg)\footnote{\leon~ is able to verify $\mem$, but it does so using a heuristic for \emph{structural induction} rather than its primary algorithm of instantiating definitions and contracts. There are other examples involving list reversal where \leon~ also requires lemmas to deal with rogue nonstandard models that fail the theorem (see Section~\ref{sec:leon}).}. We show in Section~\ref{sec:lh} that these lemmas eliminate rogue nonstandard models like the one shown above, and therefore enable the VC to be proven using unfolding techniques. There is also work in recent literature~\cite{reynolds15,Weikun19,fossil,Sivaraman2022-LemmaSynthsesisCoq} on automatically synthesizing lemmas.

\mypara{Rogue nonstandard models of integers} It is tempting to think that the above difficulties can be avoided by stating a constraint that lists are finite, i.e., there must exist a non-negative integer corresponding to the length. However, this does not work. This is because there exist \emph{rogue nonstandard models of the integers} containing elements considered `non-negative' by the model's interpretation but do not correspond to an integer (i.e., decrements do not reach $0$). The lengths of infinite lists would be interpreted to such nonstandard numbers, and we would still need user help.

\section{Preliminaries} 
\label{sec:preliminaries}

In this section we define the general setting of multi-sorted first-order logic over algebraic datatypes (ADTs) and other base types with recursively defined functions. We define \sys, our logic of study, as a fragment of this logic in Section~\ref{sec:fluid}.


\subsection{Syntax and semantics}
\label{sec:syntax-semantics}

The logic we work with is defined over a finite set of disjoint nonempty sorts $\Ss$.  
We distinguish certain sorts among these as \emph{foreground} sorts. The foreground sorts support a signature of Algebraic Datatypes (ADTs) which we describe below. The other sorts are referred to as \emph{background} sorts (background sorts could also consist of ADTs). 

An ADT signature for a sort $\sigma$ consists of a finite set of function symbols $\textit{ctor}_i, 1 \leq i \leq m$ called \emph{constructors}. Each constructor has an arity $r_i \geq 0$ and a signature $\sigma_1 \times \sigma_2 \times \ldots \times \sigma_{r_i} \rightarrow \sigma$, where $\sigma_j, \sigma \in \Ss$. 
Corresponding to each constructor with the above signature, we also have $r_i$ many \emph{destructors} $\dtor_{ij}$ with signature $\sigma \rightarrow \sigma_j$ for $1 \leq j \leq r_i$, and \emph{recognizers} $\mathit{is\_ctor}_i$ with signature $\sigma \rightarrow \mathit{Bool}$.

For example, the algebraic datatype of lists over natural numbers $\mathit{List Nat}$ is defined by the nullary constructor $\mathit{nil}: \mathit{List Nat}$ and the binary constructor $\cons: \mathit{Nat} \times \mathit{List Nat} \rightarrow \mathit{List Nat}$ whose corresponding destructors are $\head: \mathit{List Nat} \rightarrow \mathit{Nat}$ and $\tail: \mathit{List Nat} \rightarrow \mathit{List Nat}$. The recognizer $\mathit{is\_cons}$ identifies elements that correspond to non-nil lists. Note that standard pattern matching idioms for ADTs used in functional programs can be expressed using this vocabulary. 

We can also define hierarchical datatypes (e.g., lists of lists of integers), mutually recursive datatypes (e.g., terms corresponding to a context-free grammar), as well as sum (unions) and product types (tuples). We cannot define co-inductive datatypes such as infinite lists in our logic. However, we do not lose generality with respect to the various tools studied in this paper; in fact, \slh's termination checker precludes the creation of values like infinite lists.

\smallskip
\noindent Our logics have signatures of the form $\Sigma = (\Ss, \Ff, \Dd)$, where:

\begin{itemize}
    \item $\mathcal{S}$ is a finite non-empty set of sorts as defined above with a partitioning of sorts into a set of foreground ADT sorts and a set of background sorts. We require that there is at least one foreground sort.
    \item $\Ff$ is a set of constant, function, and relation symbols over the sorts $\Ss$. These will be used to model symbols over the sorts that models give interpretations to. These include 
    functions like integer addition or set union, and constructors, destructors, and recognizers over ADT sorts. 
    \item $\Dd$ is a set of function symbols distinct from $\Ff$ that will be used to model functions that have \emph{definitions}. 
\end{itemize}

The syntax is standard multi-sorted first-order logic over sorts $\Ss$ and over symbols $\Ff \cup \Dd$. We make two modifications. 
First, we require that every occurrence of a destructor term $\dtor_{ij}(t)$ is \emph{guarded} by the corresponding recognizer $is\_\ctor_i(t)$ to ensure that destructor terms are well-defined. We do not lose generality as any formula with well-defined destructor terms can be rewritten to an equivalent one with the appropriate guards. In practice, tools check that $is\_\ctor_i(t)$ holds by generating a separate Verification Condition. 
Second, we 
allow $\ite$ ($\mathit{if}\!-\!\mathit{then}\!-\!\mathit{else}$) expressions over terms and formulas. The semantics of our formulas is the standard one for FOL. We refer the reader to a standard reference text~\cite{enderton} for the notion of first-order logic, first-order models, syntax, and semantics. Semantics is defined in terms of models (aka structures) that give interpretation to all symbols, including those in $\Dd$. We use the notation $M \models \varphi$ to denote that a sentence $\varphi$ evaluates to \emph{true} in a model $M$, and 
$\varphi \models \psi$ to denote semantic entailment (all models satisfying $\varphi$ also satisfy $\psi$).


\mypara{Inductive definitions} Intuitively, a \emph{definition} of $D$ (for $D \in \Dd$) gives a particular interpretation for $D$. 
The definition of a function $D \in \Dd$ of arity $r$ is a quantified formula $\defn_D$ of the form 

\centerline{
$\forall x_1,x_2,\ldots,x_r.\, D(x_1,x_2,\ldots, x_r) = \rho(x_1,x_2,\ldots,x_r)$}    

\noindent where $\rho$ is a quantifier-free formula over $x_1$ through $x_r$ called the \emph{body} of the definition. Of course, the body may use other inductively defined symbols $G \in \Dd$. We require that every function in $\Dd$ has exactly one definition.

In order to obtain well-defined definitions, we demand a notion of termination. We define this notion using the \emph{standard model} of our logic, which we introduce in the next section.









\subsection{The standard model}
\label{sec:standard-model}



The intended standard interpretation of an ADT signature is the initial term algebra where the universe consists of terms that respect the sorts and the interpretation of constructors is that of term application, i.e., $\llbracket\textit{ctor}_i\rrbracket(e_1, \ldots, e_{r_i}) = \textit{ctor}_i(e_1, \ldots, e_{r_i})$.


The destructors are interpreted as $\llbracket\dtor_{ij}\rrbracket(\ctor_i(e_1, \ldots, e_{r_i})) = e_j$
and is otherwise interpreted to be a default value on other elements~\footnote{Since we consider only formulas that are guarded to check elements to be of the right sort before applying destructors, the semantics of the formula on other elements is irrelevant.}. Finally, recognizers are only true on terms constructed with the corresponding constructor: $\llbracket\mathit{is\_ctor}_i\rrbracket(\ctor_i(e_1,e_2,\ldots e_{r_i}))  = \mathit{True}$, and is $\mathit{False}$ for other elements.

More generally, our logic is parameterized by a \emph{standard model} $\Mm_{\Ss,\Ff}$ of the foreground and background sorts. 
This is typically true of sorts used in program verification: ADTs, integers, sets, etc. 
Note that this model does not give interpretations to functions in $\Dd$.


We require that inductive definitions are terminating on the standard model using a standard \emph{eager} semantics~\cite{WinskellSemantics}. Informally, we evaluate a definition on concrete elements over the standard model as follows: (i) we evaluate recursively defined function terms by evaluating the definition on the arguments; (ii) for $\ite$ expressions, we evaluate the conditions first and then only evaluate the appropriate branch; (iii) for all other expressions, we first evaluate all recursively defined function terms (with subterms evaluated before their superterms) and then evaluate the expression. A terminating definition is one for which this procedure terminates on all inputs. 

The following proposition states that over $\Mm_{\Ss,\Ff}$, there exists a unique valuation for the defined functions $\Dd$ that is consistent with their definition.

\begin{proposition}
\label{prop:standard-model-unique-defs}
Given $(\Ss, \Ff, \Dd)$ let $\defns =\{\defn_D \,\mid\, D \in \Dd\}$ be a set of definitions. There exists a unique model $\Mm_{\Ss,\Ff,\Dd}$ such that the interpretation of symbols in $\Ff$ coincides with $\Mm_{\Ss,\Ff}$ and interpretations of symbols in $D \in \Dd$ satisfy their definitions, i.e., $\Mm_{\Ss,\Ff,\Dd} \models \defn_D$. 
\end{proposition}

Observe that functional programs can be modeled using definitions (we only consider terminating programs, of course). Then, verifying universal FO properties of functional programs, say $\varphi$, can be modeled as checking the validity of formulas of the form $\defns \rightarrow \varphi$. As a slight abuse of notation we use $\defns$ in formulas to mean the conjunction of formulas in the set of definitions $\defns =\{\defn_D \,\mid\, d \in \Dd\}$. 

However, as discussed in Section~\ref{sec:introduction} it is easy to show that the problem of validity of even quantifier-free formulas on $M_{\Ss,\Ff,\Dd}$ is \emph{not recursively enumerable}.

\begin{proposition}[Incompleteness Theorem for the Standard Model]
\label{prop:standard-model-incompleteness}
There exists a sort $\sigma$ with an ADT signature $\Ff$ and defined functions $\Dd$ such that checking $\Mm_{\{\sigma\},\Ff,\Dd} \models \varphi$ is not recursively enumerable for quantifier-free $\varphi$.
\end{proposition}

Note that validity over ADTs without background sorts and definitions is decidable~\cite{malcev62} since it has a complete axiomatization~\cite{hodges97}. The introduction of definitions (programs) is what leads to incompleteness.

\subsection{Combinations of theories, nonstandard models, and rogue nonstandard models}
\label{sec:combined-theories}

A primary observation we make in this paper is that techniques for reasoning based on function unfolding and quantifier-free reasoning (as in {\sc Liquid Haskell} and {\sc Leon}) do not reason with the standard model but rather with a certain \emph{combination of first-order theories}. We will show this in Section~\ref{sec:soundness}. In this section we formalize the notation for combined theories. 

A \emph{theory} $\Tt$ for a signature is an entailment closed set of first-order sentences. A model $\Mm$ satisfies a theory $\Tt$, denoted $\Mm \models \Tt$, if every sentence in the theory holds in the model. A sentence $\psi$ is valid in $\Tt$, denoted $\Tt \models \psi$ if $\psi$ belongs to $\Tt$.

A \emph{theory tuple} for $(\Ss, \Ff, \Dd)$ is:
\begin{itemize}
    \item The first-order theory of ADTs $\Tt_\sigma$ for each foreground sort $\sigma$. This is precisely the theory of the standard ADT model 
    for $\sigma$, which may involve functions over other sorts 
    using which the elements of $\sigma$ are to be constructed. These other sorts 
    are themselves constrained by theories like Presburger Arithmetic, or an ADT theory. 
    \item A theory $\Tt_\textit{bg}$ for the combined signature of the background sorts that is recursively enumerable. We require the background sorts in the standard model to satisfy this theory. In practice, this theory is the union of several axiomatized theories, say for arrays, integers, bitvectors, etc.
    \item Theory of uninterpreted functions with equality for symbols in $\Dd$ (i.e., the so-called empty theory).
\end{itemize}

The combined theory $\thcomb$ of a theory tuple is the entailment closure of the union of the theories in the tuple. A model satisfies a  theory tuple (and consequently the combined theory) if the projection of the model to each subset of sorts satisfies the theories constraining those sorts. The combined theory $\thcomb$ is the set of all FO sentences that hold in all these models. For example, consider the ADT $\mathit{List Nat}$ of lists over natural numbers introduced earlier. A theory tuple for this signature could be one that has (a) the theory of ADT lists for the foreground sort, and (b) the theory of Presburger Arithmetic (natural numbers with addition) for the background sort. The combined theory is the entailment closure of the union of the two theories.

Note that the first order theory of ADTs is \emph{complete}. Therefore, the above setting is agnostic to the choice of any complete axiomatization for the ADT sorts!~\cite{hodges97,kovacs17}. Consequently, our results are also quite general and agnostic to the choice of axiomatization.

The standard models for each sort satisfy their respective theories. The other models of the individual theories are called \emph{nonstandard models}. The combined standard model $\Mm_{\Ss,\Ff,\Dd}$ is a model of $\thcomb$, and similarly the other models of $\thcomb$ are nonstandard models for the combined theory.

Since the standard model is a model of $\thcomb$, it is clearly the case that $\thcomb$ is a \emph{subset} of the theory of the standard model $\Mm_{\Ss,\Ff,\Dd}$, which we denote by $\thstd$. However, the reverse is not true in general, and in fact the combined theory can be \emph{strictly smaller} than the theory of the standard model. For example, consider the above example of $\mathit{List Nat}$ where we extend the logic with the predicate symbol $R$ with the following recursive definition:
\begin{align*}
  R(x) = \mathit{ite}&(\mathit{is\_nil}(x), \mathit{False},\\ \mathit{ite}&(\mathit{is\_nil}(\tail(x)),\head(x)=1,\\
  &\hspace{65pt}\head(x)=\head(\tail(x)) \land R(\tail(x))))
\end{align*}
%
%
One would expect that $R(x)$ holds only for nonempty lists $x$ whose elements are all $1$. Indeed, the statement $R(x) \rightarrow \head(x)=1$ is valid on the standard model. However, this sentence is \emph{not valid} in the combined theory as there is a \emph{rogue nonstandard model} that falsifies it. In this work, we define a rogue nonstandard model as a nonstandard model that falsifies a theorem of interest which is valid on the standard model. 

An example of a rogue nonstandard model falsifying $R(x) \rightarrow \head(x) = 1$ is as follows. It has an element $u$ in the ADT universe that does not correspond to any standard (i.e., finite) ADT term such that $R(u)$ is true and $\head(u)=2$. Destructing $u$ consecutively would proceed forever without reaching $\mathit{nil}$ and and all these elements will satisfy $R$ and have their head element to be $2$, hence satisfying the recursive equation for $R$. 


Standard and nonstandard models satisfy the same FO properties for ADTs, but the addition of recursively defined functions destroys this. Although the combination of ADTs and recursively defined functions is the primary technical hurdle, 
we develop completeness results for a theory that also includes background sorts. This is crucial to verify functional programs as they invariably involve background theories.





Formally, in this paper we work with validity under the combination of theories $\thcomb$. 

\subsection{Validity under defined functions}
\label{sec:validity}

%
Let $\defns$ be a set of definitions $\defn_D$ for each $D \in \Dd$. We define the validity of a first-order formula $\varphi$ under definitions $\defns$ and modulo an arbitrary theory $\Tt$ by considering pairs of the form $(\defns, \varphi)$. 
Recall that when we use $\defns$ in formulas we mean the conjunction of definitions in the set.

\begin{definition}[Validity of FOL Formulae with Defined Functions]
\label{def:defs-phi-validity}
Given $(\Ss, \Ff, \Dd)$ along with definitions $\defns$ of functions in $\Dd$ and an FO-theory $\Tt$, we say that a formula $\varphi$ is $\Tt$-valid under the definitions iff 
$\defns \rightarrow \varphi$ is $\Tt$-valid, i.e., is in $\Tt$.  
We denote this by $\Tt \models (\defns, \varphi)$.\qed
\end{definition}

We can utilize the above notion in the case of the theory of the standard model or the combined theories, writing $\thstd \models (\defns,\varphi)$ or $\thcomb \models (\defns,\varphi)$ respectively. As before, if  $\thcomb \models (\defns,\varphi)$ then $\thstd \models (\defns,\varphi)$.



\section{A \sys logic}
\label{sec:fluid}





In this section we define our first main contribution: the \sys (First-Order Logic of Universal properties under Inductive Definitions) fragment that captures VCs generated by tools like \slh and \leon. 
The heart of the \sys fragment is a class of inductive definitions called \emph{provably acyclic} definitions.

Recall that we require definitions to terminate on the standard model. We demand in the \sys fragment that definitions also satisfy a \emph{provable acyclicity} condition, which is a notion similar to termination. 
Intuitively, acyclicity means that when definitions are unrolled, there is no cyclic dependency between the recursive calls. Note terminating functions must be acyclic, but acyclic definitions can be non-terminating. For example, the function $\mathit{forever}$ on Lists defined by $\mathit{forever}(x) = \mathit{forever}(\cons(0,x))$
does not terminate, but it is acyclic because the recursive calls do not repeat. We demand in the \sys fragment that the acyclicity property expressed as a first-order formula is \emph{provable} for the recursive definitions\footnote{A subtle point here is that 
even though terminating functions are acyclic, 
they need not be \emph{provably} acyclic 
\iffull
(see the \hyperlink{acyclicity-vs-termination}{discussion} at the end of this section 
\fi
for an example). Therefore, 
termination and provable acyclicity are incomparable. Of course, \emph{provable} termination implies provable acyclicity, as expected.}. We formulate provable acyclicity below using ranking functions and path conditions, which we first define formally.

\smallskip
An ordered sort $S \in \Ss$ is one with a binary predicate $<$ such that $<$ forms a strict partial order. Formally, $<$ must satisfy the FO axioms expressing irreflexivity and transitivity under $\thcomb$. Note that $<$ need not be well-founded because we only require acyclicity, not termination\footnote{Well-foundedness is not expressible in FOL anyway.}. Every ADT sort is an ordered sort with respect to the (strict) subterm relationship.

For a recursively defined function $D \in \Dd$, a \emph{ranking function} for $D$ is a recursively defined function $\mathit{Rank}_D \in \Dd$ from the domain of $D$ to an ordered sort. We require that $\Dd$ is \emph{stratified}. The stratum of a function $D \in \Dd$ is a natural number denoted by $\strat(D)$. Note that multiple functions can have the same strata. We require that every $D \in \Dd$ with $\strat(D) > 0$ has a ranking function $\mathit{Rank}_D$ whose stratum is \emph{strictly lower} than $D$. When $\strat(D) = 0$, we require that $D$ is unary over an ordered sort, and its ranking function is the identity function. Finally, we require that the definition of a function at stratum $i$ can only call functions of strata lower than or equal to $i$.




\smallskip
We now define path conditions.

\begin{definition}[Path Condition]
\label{def:pathcond}
Given a formula $\rho$\footnote{$\rho$ is usually the body of a recursively defined function}, we denote by $\pathcond_\rho(\psi,E)$ that the sub-expression (subterm or subformula) $E$ occurs in $\rho$ with path condition $\psi$. It is the least relation satisfying the following recurrence:
\begin{align*}
    \pathcond_\rho(\mathit{True}, \rho) &\, \textrm{ holds}\\
    \textrm{If  } \pathcond_\rho(\psi, \ite(\mathit{cond}, E_1, E_2)) &\, \textrm{ then  } \pathcond_\rho(\psi, \mathit{cond})\\
    \textrm{If  } \pathcond_\rho(\psi, \ite(\mathit{cond}, E_1, E_2)) &\, \textrm{ then  } \pathcond_\rho(\psi \land \mathit{cond}, E_1)\\
    \textrm{If  } \pathcond_\rho(\psi, \ite(\mathit{cond}, E_1, E_2)) &\, \textrm{ then  } \pathcond_\rho(\psi \land \neg\mathit{cond}, E_2) \\
    \textrm{If  } \pathcond_\rho(\psi, D(E_1\ldots,E_n)) \textrm{ for $D \in \Dd$} &\, \textrm{ then } \pathcond_\rho(\psi, E_j), 1 \leq j \leq n\\
    \textrm{If  } \pathcond_\rho(\psi, \oplus(E_1\ldots,E_n)) \textrm{ for $\oplus \neq \ite, \oplus \notin \Dd$} &\, \textrm{ then } \pathcond_\rho(\psi, E_j), 1 \leq j \leq n
\end{align*}
\end{definition}

Informally, the path condition is the conjunction of all the conditions of $\ite$ expressions that must be satisfied in order to ``reach'' the given sub-expression.

\medskip
\noindent
We are now ready to define provable acyclicity:

\begin{definition}[Provably Acyclic Definitions]
\label{def:provably-acyclic}
Given a signature with combined theory $\thcomb$ and stratified definitions $\defns$, a definition $\defn_D \equiv \forall \overline{x}.\, D(\overline{x}) = \rho(\overline{x})$ is provably acyclic if for every $G(\overline{t})$ ($G \in \Dd$) occurring in $\rho$  with $\strat(G) = \strat(D)$, $\mathit{Rank}_G$ and $\mathit{Rank}_D$ have the same range sort, and furthermore, for every $\psi$ such that $\pathcond_\rho(\psi, G(\overline{t})))$:

\centerline{$\thcomb \models \left(\bigwedge\limits_{\strat(H) < \strat(D)}\defn_H\right) \rightarrow \left(\psi \,\rightarrow\, \mathit{Rank}_G(\overline{t}) \,\!<\!\, \mathit{Rank}_D(\overline{x})\right)$}

\noindent
where the overloaded symbol $<$ represents an order predicate in the corresponding sort.\qed
\end{definition}

Informally, the above definition says that the arguments to recursive calls must be \emph{provably} (w.r.t $\thcomb$) smaller than the input arguments as measured using ranking functions. Although we say `provable', note that the definition uses semantic entailment ($\models$). However, these two notions are the same since FOL is complete. Provable acyclicity ensures that when a definition is unrolled, there is no cyclic dependency between recursive calls as the arguments will always decrease. We can use the definitions of functions in lower strata and the path condition to the recursive call to establish this property. We give an example below.

\begin{example}[Sorted List Merge]
Consider the usual function $\mathit{merge}(x, y)$ for merging sorted lists:

\rev{
$\forall x,y:\mathsf{List}.\, {\mathit{merge}}(x,y) = \;\ite(\isnil(x),y,$

$\textrm{\hspace{10.3em}\;\;} \ite(\isnil(y),x,$

$\textrm{\hspace{10.3em}\;\;} \ite(\head(x) \leq \head(y),\cons(\head(x),\mathit{merge}(\tail(x),y)),$

$\textrm{\hspace{20em}\;\;}\cons(\head(y),\mathit{merge}(x,\tail(y))))))$
}

Let the stratum for $\mathit{merge}$ be $1$, with its ranking function being the sum of lengths of $x$ and $y$. The stratum of the length function $\mathit{length}$ is $0$. 

Consider the recursive call $\mathit{merge}(\tail(x),y)$. The path condition in this case is $x \neq \nil \land y \neq \nil \land \head(x) \leq \head(y)$. 
We can show that this call has smaller rank, i.e., $(\mathit{length}(\tail(x)) + \mathit{length}(y)) < (\mathit{length}(x) + \mathit{length}(y))$ using the definition of $\mathit{length}$ 
and the path condition ($x \neq \nil$ ensures that the term $\tail(x)$ is well-defined). We can also show similarly that the other recursive call has smaller arguments, and therefore $\mathit{merge}$ is provably acyclic.

We can also show that $\mathit{length}$ is provably acyclic. Since its stratum is 0, its ranking function must be identity, therefore we have to show that the arguments to recursive calls must themselves decrease. This is evidently true since $\mathit{length}(x)$ recurses on $\tail(x)$, which is smaller according to the ADT subterm ordering.\qed
\end{example}


\mypara{Aside} We note here some subtleties in the 
definition of provable acyclicity. First, the relation $<$ is a \emph{mathematical} one, and does not need to be part of the signature or logically defined. Consequently, Definition~\ref{def:provably-acyclic} can be established by a user/system in any way. For example, if $<$ denotes the subterm ordering on ADT Lists, then a system can trivially deduce that $\tail(x) < x$. In particular, a definition that recurses on destructions of the called arguments is immediately provably acyclic. Second, ADTs are an ordered sort regardless of the choice of axiomatization because the subterm relation is an order in any model that satisfies a complete axiomatization of ADTs, including nonstandard models (even rogue ones). Therefore, we do not need ADT signatures/axiomatizations with an explicit subterm predicate~\cite{kovacs17}. 
Third, observe that ranks 
need not be well-founded as we only require acyclicity, not termination. Contrary to the usual ranking functions in literature, ranks need not be lower-bounded. For example, the function $\mathit{forever}$ defined above is provably acyclic because we can say $\mathit{Rank}(\cons(0,x)) < \mathit{Rank}(x)$ with the rank being negative of the length, which has no lower bound.

\smallskip
Our fragment is very general and includes most definitions we know that tools use. In practice, provable acyclicity is satisfied when functional programs are proved terminating. This is because, to the best of our knowledge, every tool that proves functional programs terminating uses ranking functions that map arguments to a well-founded order (typically tuples of natural numbers, often associated with the size of ADTs), and shows that (1) the ranking function decreases (according to some order relation $<$ ) on recursive calls, and (2) the order $<$ on which the ranking function decreases is well-founded. Condition (1) is precisely the property in Definition~\ref{def:provably-acyclic}!


Intuitively, provable acyclicity generalizes the idea of proving termination. Termination makes sense on a standard model, but in a nonstandard model ADT elements can have ``infinite tails'' and therefore a function that terminates on the standard model can be nonterminating on a nonstandard model\footnote{Here we mean nonterminating in the sense that the evaluation procedure described in Section~\ref{sec:standard-model} does not terminate for all inputs drawn from the nonstandard model.}. In contrast, provable acyclicity makes sense on all models, standard and nonstandard. We show that in any $\thcomb$ model, provably acyclic definitions are always satisfiable (though they may not have a unique interpretation). Formally 
\iffull
(see Corollary~\ref{cor:provably-acyclic} in Section~\ref{sec:completeness} 
\fi
for a proof):

\begin{theorem}
\label{thm:provably-acyclic}
Given a signature $(\Ss, \Ff, \Dd, \thcomb)$, a set of stratified definitions $\defns$ that are provably acyclic, and a model $\Mm$ of $\thcomb$, there exists a model $\Mm'$ of $\thcomb$ such that the interpretation of symbols in $\Ff$ coincides with $\Mm$ and interpretations of symbols in $\Dd$ satisfy their definitions.
\end{theorem}


\smallskip
\mypara{\sys} We now define the \sys fragment.

\begin{definition}[\sys Fragment]
\label{def:safe-fragment}
Given a signature $(\Ss, \Ff,\Dd,\thcomb)$ and a set of stratified definitions $\defns$ for the symbols in $\Dd$, the pair $(\defns, \varphi)$ is in the \sys fragment if (1) every definition in $\defns$ is provably acyclic, and (2) $\varphi$ is purely universally quantified.
\end{definition}

\added{

\mypara{Discussion on provable acyclicity vs. termination}\hypertarget{acyclicity-vs-termination}{}

Termination refers only to termination on the standard model. In contrast, provable acyclicity requires that arguments do not repeat when unfolding a definition (on any model), which is established using an order predicate and ranking functions. However, neither one implies the other.

For example, the function $f(x) = f(cons(0,x))$ is a provably acyclic function since the arguments to the function will never repeat across successive recursive calls. However, $f$ does not terminate on the standard model: it simply calls itself on larger and larger lists.

In contrast, the following predicate $g$ terminates on the standard model but is not provably acyclic:

\begin{center}
$\forall x.\, \mathit{std}(x) = \ite(x = \nil, \mathit{True}, \mathit{std}(\tail(x)))$\\
$\forall x.\, g(x) = \ite(\mathit{std}(x), \mathit{True}, g(x))$
\end{center}

$\mathit{std}$ always terminates on the standard model returning $True$: it continually destructs the element and recursively calls itself until it reaches $\nil$, at which point it returns $\mathit{True}$. Therefore, $g(x)$ also terminates for elements in the standard model as one would simply evaluate the outermost condition (which terminates), and then take the branch corresponding to the success of the condition (which is just $\mathit{True}$).

However, $g$ is not provably acyclic: it recursively calls $g(x)$ which does not decrease the argument. We can’t use the fact we used in the termination argument that the $\mathit{else}$ branch will never be taken because for provability we have to consider all models, not just the standard model. There are models where $\mathit{std}$ does not always evaluate to $\mathit{True}$.

In Section~\ref{sec:boundarythms} we use the above construction 
and the fact that $\mathit{std}$ cannot be proved to always evaluate to $\mathit{True}$ to show incompleteness of \qialg\ when provably acyclic functions are replaced by terminating functions (see Theorem~\ref{thm:boundary-incompleteness}).

} 
\section{Completeness of recursive function unfolding and quantifier-free reasoning}
\label{sec:qialg}

In this section we describe the algorithm $\qialg$, based on Unfolding definitions followed by Quantifier-Free Reasoning, for checking validity of universal properties. We show that the algorithm intrinsically only proves theorems in the combined theory $\thcomb$. We then prove the main technical contribution of this paper, that the algorithm is \emph{complete} for $\thcomb$. Let us fix a signature $(\Ss, \Ff, \Dd, \thcomb)$ through this section. Recall that $\thcomb$ represents the combined theory for the foreground and background sorts, with $\Dd$ being uninterpreted. Fix also the theory of the standard model $\thstd$.

We require for our algorithm that $\thcomb$-validity is \emph{decidable} for \emph{quantifier-free formulas}, and that the quantifier-free fragments of $\thcomb$ and $\thstd$ are identical. We are agnostic to the choice or presence of an axiomatization for the theories and have no other constraints on $\thcomb$. This assumption is satisfied for several combined theories, including those that admit Nelson-Oppen combination~\cite{nelson-oppen1979, tinnelizarba} \eg ADTs, linear arithmetic, reals, etc. In fact, such theories also admit efficient decision procedures as evidenced by SMT solvers~\cite{z3,cvc4}. Checking validity is achieved by negating and checking for unsatisfiability. Note that the \emph{quantified} theories $\thstd$ and $\thcomb$ are however typically different; see Section~\ref{sec:combined-theories}.


\subsection{\qialg~algorithm}
\label{sec:qialg-description}

The high-level picture of the algorithm is as follows: presented with a set of definitions $\defns$ and a quantifier-free formula $\varphi$, \qialg~ systematically unfolds the definitions on terms on which functions are applied and dispatches the resulting quantifier-free formulas to a decision procedure for satisfiability.
We first provide some  definitions that are useful in describing the algorithm.

\begin{definition}[$\Dd$-Application]
\label{def:recdef-terms}
A $\Dd$-application is a pair $(D,\tup)$ for $D \in \Dd$ and a tuple of terms $\tup = (t_1,t_2\ldots, t_r)$ such that $D(\tup)$ is well-formed, i.e., $D$ has signature $\sigma_1 \times \sigma_2 \ldots \times \sigma_r \rightarrow \sigma$ and $t_i$ is of type $\sigma_i$ for $1 \leq i \leq r$.
A $\Dd$-application $(D,\tup)$ occurs in a formula $\psi$ if $D(\tup)$ occurs in $\psi$.\qed
\end{definition}

\begin{definition}[Definition Unfolding]
\label{def:definition-unfolding}
Let $\varphi \equiv \forall x_1.\,\forall x_2.\,\ldots\forall x_n.\, \psi$ be a universally quantified formula such that $\psi$ is quantifier free. The instantiation of $\varphi$ with a tuple of terms $\tup \equiv (u_1, u_2, \ldots u_n)$, written $\varphi[\tup]$, is the quantifier-free formula $\psi[u_1/x_1, \ldots u_n/x_n]$.\qed
\end{definition}

Given a set of $C$ of $\Dd$-applications and a set $\defns = \{\defn_D \mid D \in \Dd\}$ 
of definitions 
we denote $\defns[C] = \{\defn_D[\tup] \mid (D,\tup) \in C\}$. Informally, $\defns[C]$ is the set of quantifier-free formulas that contains all the unfoldings of definitions corresponding to the pairs in $C$.



\begin{algorithm}[t]
\begin{algorithmic}[1]
\Statex {\bf Input:} $(\defns, \varphi)$ such that $\varphi$ is universally quantified, with $\defns = \{\defn_D \mid D \in \Dd\}$
\Statex {\bf Output:} $\mathit{VALID}$ (when it terminates)
\Statex {\bf Imports:} $\smtproc$ for deciding $\thstd$-satisfiability of quantifier-free formulas
\Procedure{\textsc{{\bfseries \qialg}}$[\Ss; \Ff; \Dd; \thstd]$}{}
\State $\formulas := \{\neg \varphi\}$ \CodeComment{Negate the formula and Skolemize}
\While{$\mathit{True}$}
\label{qialg:outer-loop}
\State $\mathit{res} := \smtproc(\formulas)$ \CodeComment{Check sat of $\neg\varphi$ with current unfoldings}
\label{qialg:checksat}
\If{$\mathit{res} = \mathit{UNSAT}$} 
\State \Return{$\mathit{VALID}$} \CodeComment{$(\defns,\varphi)$ is valid}
\Else
\State \CodeComment[\linewidth]{Compute $\Dd$-applications occurring in $\formulas$}
\State $\unfolded := \{ (D,\tup) \mid D(\tup) \textrm{ occurs in } \psi \textrm{ for } \psi \in \formulas\}$
\label{qialg:get-terms}
\State \CodeComment[\linewidth]{Unfold the definitions and add them to formulas}
\State $\formulas := \formulas \cup \defns[\unfolded]$
\label{qialg:add-instantiation}
\EndIf
\EndWhile
\EndProcedure
\end{algorithmic}
\caption{Algorithm for Unfolding Definitions followed by Quantifier-Free Reasoning}
\label{fig:qialg}
\end{algorithm}

\mypara{Algorithm Description} Algorithm~\ref{fig:qialg} shows the pseudocode for \qialg, parameterized by the signature $(\Ss, \Ff,\Dd,\thstd)$ with the theory of the standard model. It takes as input a set of definitions $\defns = \{\defn_D \mid D \ \in \Dd \}$ and a formula $\varphi$ such that $\varphi$ is universally quantified. 
\qialg~ attempts to prove validity by establishing unsatisfiability of the negation $\defns \land \neg\varphi$. 
Finally, \qialg\ uses an external procedure $\smtproc$ that checks the satisfiability of quantifier-free formulas with respect to 
$\thstd$. It takes as input a set of formulas and outputs $\mathit{SAT}$ if the conjunction of the formulas is $\thstd$-satisfiable and $\mathit{UNSAT}$ otherwise. 

The algorithm maintains a set $\formulas$ of quantifier-free formulas consisting of $\neg\varphi$ along with finitely many unfoldings of the definitions. If this set is unsatisfiable then the formula $\defns \land \neg\varphi$ is unsatisfiable as well, i.e., $\varphi$ is valid under $\defns$. Initially the set contains only $\neg\varphi$. Since $\varphi$ is purely universal, we Skolemize the quantified variables in $\neg\varphi$ and treat it as a quantifier-free formula, adding the existentially quantified variables as new ground terms (constants) in our signature. 

At a general point in the algorithm (line~\ref{qialg:outer-loop}), we check the $\thstd$-satisfiability of $\formulas$ using the external procedure $\smtproc$ (line~\ref{qialg:checksat}). Note that although there is a unique valuation for every $D \in \Dd$ on the standard model consistent with $\defns$, the set $\formulas$ only enforces this consistency for finitely many unfoldings of $\defns$ and otherwise treats the symbols in $\Dd$ as uninterpreted, which is an over-approximation. 
If $\formulas$ is unsatisfiable we exit and return $\mathit{VALID}$. Otherwise, we refine our approximation by unfolding the definitions on more terms. We compute the set of $\Dd$-application terms occurring in $\formulas$ (line~\ref{qialg:get-terms}), add the corresponding unfoldings of definitions to $\formulas$ (line~\ref{qialg:add-instantiation}), and go back to the beginning of the loop. 
Observe that if the algorithm is not able to prove the unsatisfiability of $\neg\varphi$ using any amount of unfoldings then it does not terminate. 


\subsection{Soundness and completeness of \qialg~ under combined theories}
\label{sec:soundness}

In this section we prove the primary contribution of this work, namely that \qialg~ is complete for $\thcomb$-validity of \sys formulas. We first show  the soundness of \qialg.

\begin{theorem}[Soundness of \qialg~ w.r.t $\thcomb$]
\label{thm:soundness}
If $\qialg(\Ss; \Ff; \Dd; \thstd)$ terminates on $(\defns, \varphi)$ then $\thcomb \models (\defns, \varphi)$.
\end{theorem}


\begin{proof}
In each round of the algorithm the set $\formulas$ is of the form $\defns[C] \cup \{\neg\varphi\}$, where $\defns[C]$ contains unfoldings (i.e., instantiations) of $\defns$ on a set $C$ of $\Dd$-applications. Therefore, if \qialg~ terminates then $\defns[C] \land \neg\varphi$ is unsatisfiable with respect to $\thstd$ (line~\ref{qialg:checksat}). 

Now, $\smtproc$ can also be seen as a satisfiability procedure for the combined theory $\thcomb$ since the input formulas are \emph{quantifier-free}. We hence have that $\defns[C] \land \neg\varphi$ is unsatisfiable with respect to $\thcomb$, which yields $\defns \land \neg\varphi$ is unsatisfiable with respect to $\thcomb$, \ie, $\thcomb \models (\defns,\varphi)$.
\end{proof}

We showed in Section~\ref{sec:combined-theories} that typically $\thstd$ is strictly larger than $\thcomb$. The above result shows that the proving power of \qialg~ is in fact bounded by $\thcomb$. 
Therefore, not only are there valid theorems in $\thstd$ that are not valid in $\thcomb$, but it is also the case that \qialg~ (and hence systems such as \slh and \leon) will never be able to prove those theorems. 


\label{sec:completeness}
\medskip
We now show that \qialg~ is complete. 

\begin{theorem}[Completeness of \qialg~ w.r.t $\thcomb$ for \sys]
\label{thm:completeness}
If $(\defns, \varphi)$ belongs to the \sys fragment (Definition~\ref{def:safe-fragment}) and $\thcomb \models (\defns, \varphi)$, then $\qialg(\Ss; \Ff; \Dd; \thstd)$ terminates on $(\defns, \varphi)$ and reports it valid.
\end{theorem}

We dedicate the rest of this section to the proof of the completeness theorem.

\subsection*{Prologue: Theorem simplification and reduction to model construction}
We make some simplifications for ease of presentation. First, we assume that $\defns$ has only one stratum. 
We provide a generalization of the argument made here to multiple strata at the end of this section. Second, we assume without loss of generality that the signature $(\Ss, \Ff, \Dd, \thcomb)$ is such that if a formula $\Gamma$ is satisfiable in a $\thcomb$ model, then it is satisfiable in a Herbrand model consisting of the terms occurring in $\Gamma$ and closed under the applications of functions in $\Ff \cup \Dd$. This can always be done by Skolemizing $\thcomb$ and expanding $\Ff$ with new function symbols.

We first rewrite the statement of the theorem to an equivalent one. Consider the value of the sets $\formulas$ and $\unfolded$ through the algorithm:
\begin{align*}
    \formulas_0 &= \{\neg\varphi\} \tag{initial value}\\
    \unfolded_i &= \{(D, \tup) \mid D(\tup) \textrm{ occurs in } \psi \in  \formulas_{i-1}\} \tag{$i > 0$}\\
    \formulas_i &= \formulas_{i-1} \cup \defns[\unfolded_i] \tag{$i > 0$}
\end{align*}

\noindent
where the subscript $i$ denotes their values in the $i^{th}$ \emph{round} of the outermost loop on line~\ref{qialg:outer-loop}. Observe that $\formulas_i \subseteq \formulas_j$ and $\unfolded_i \subseteq \unfolded_j$ for every $j > i$. The completeness result can then be stated as follows:

\begin{theorem}[Completeness of \qialg~ w.r.t $\thcomb$]
\label{thm:completeness-formal}
If $\thcomb \models (\defns, \varphi)$ then $\formulas_i$ is $\thcomb$-unsatisfiable for some $i \geq 0$.
\end{theorem}

Note that the above theorem implies that \qialg~ is complete for $\thcomb$-validity because if for some $i$ we have that $\formulas_i$ is $\thcomb$-unsatisfiable, then it is also $\thstd$-unsatisfiable, therefore the algorithm will terminate in round $i$. By the soundness theorem (Theorem~\ref{sec:soundness}), $(\defns, \varphi)$ is $\thcomb$-valid.

We show the contrapositive of the above statement. Let us assume that $\formulas_i$ is $\thcomb$-satisfiable for every $i \in \mathbb{N}$. We show that $\defns \land \neg\varphi$ is $\thcomb$-satisfiable. Specifically, we construct a $\thcomb$ model $\Nn$ such that $\Nn \models \defns \land \neg\varphi$.

\mypara{Proof Plan} 
We construct $\Nn$ in two stages:
\begin{enumerate}[leftmargin=36pt]
    \item[Stage 1:] We first use the assumption that $\formulas_i$ is $\thcomb$-satisfiable for every $i \in \mathbb{N}$ to construct a model $\Mm$ (using compactness) that satisfies $\bigcup\limits_{i \geq 1}\defns[\unfolded_i]$ and $\neg\varphi$. Note that this model need not satisfy $\defns$ \emph{everywhere} (as we have only instantiated definitions for a subset of terms). 
    \item[Stage 2:] In this stage we take the model $\Mm$ and consider a finite set $K$ of pairs of the form $(D, \tup)$ such that the interpretation of $D$ in $\Mm$ does not satisfy the definition of $D$ on $\tup$. We show that we can `repair' the model so that definition of $D$ now holds on $\tup$ for every $(D, \tup) \in K$. We then show that definitions can be repaired everywhere using a compactness argument. This results in the model $\Nn$ we seek.
\end{enumerate}

\subsection*{Stage 1: Model of infinite instantiations}

We recall the compactness theorem for FOL under combinations of theories.

\begin{proposition}[FOL Compactness with Theories]
\label{prop:fol-compactness}
Given a signature $(\Ss, \Ff, \Dd, \thcomb)$, a set of formulas $\Gamma$ (finite or infinite) is $\thcomb$-satisfiable if and only if every finite subset of $\Gamma$ is $\thcomb$-satisfiable.\qed
\end{proposition}

From our assumption we know that $\formulas_i$ is $\thcomb$-satisfiable for every $i$. Using compactness and the fact that $\formulas_i$ form an increasing sequence w.r.t $\subseteq$, it follows
that the infinite set $\infunfold = \bigcup\limits_{i \in \mathbb{N}}\formulas_i$
is $\thcomb$-satisfiable. We rewrite this as $\infunfold = \{\neg\varphi\} \cup \bigcup\limits_{i \geq 1}\defns[\unfolded_i]$. 

\smallskip
\noindent
Let $\Mm$ be a $\thcomb$-model that satisfies $\infunfold$. From our simplifying assumptions, we can assume that $\Mm$ is a Herbrand model. It satisfies $\neg\varphi$ and satisfies the definitions only on certain tuples, namely for $(D, \tup) \in \bigcup\limits_{i \geq 1} \unfolded_i$. 

Note here that if the model $\Mm$ happened to be the \emph{standard model} the repair we wish to do would be trivial as $\defn_D$ is uniquely defined (see Proposition~\ref{prop:standard-model-unique-defs}) for each $D \in \Dd$ and we can simply `complete' the model with the correct valuations. However, $\Mm$ can be a nonstandard model, and this results in the nontrivial aspects of our construction below.

\subsection*{Stage 2: Computational closure and model repair}
The reason we can repair $\Mm$ 
is because the set $\bigcup\limits_{i \geq 1} \unfolded_i$ has a special property: it is \emph{computationally closed}. We define this property below.

\begin{definition}[Computationally Closed Set]
\label{def:computational-closure}
Let $\Gamma$ be a set of quantifier-free formulas. 
A set $C$ of $\Dd$-applications is said to be computationally closed with respect to $\Gamma$ if: (1) if $D(\tup)$ occurs in some formula in $\Gamma$ then $(D, \tup) \in C$, and (2) if $(D, \tup_1) \in C$ and a $\Dd$-application $(G, \tup_2)$ occurs in $\defn_D[\tup_1]$ then $(G, \tup_2) \in C$. 
\end{definition}

Intuitively, for a recursively defined function $D$, the computational closure of a term $D(t)$ contains all the recursive calls (at any level) made by a call to $D$ on $t$, where we represent a recursive call to a function $G$ on a term $r$ by the $\Dd$-application $(G, r)$. The set is called a computational closure because it is the set of calls that occur when `computing' the value of $D$ on $t$ symbolically. The computational closure of a formula is then the union of the computational closures of all terms of the form $D(t)$ occurring in the formula. For example, consider the length function $\mathit{length}$ on Lists. The computational closure of $\mathit{length}(\cons(1,\nil))$ is the set $\{(\mathit{length}, \cons(1,\nil)),\, (\mathit{length}, \nil)\}$. Similarly, the computational closure of $\mathit{length}(x)$ is $\{(\mathit{length},x),\, (\mathit{length},\tail(x)),\, (\mathit{length},\tail(\tail(x))),\, \ldots\}$.


Using the above definition we can see that $\bigcup\limits_{i \geq 1}\unfolded_i$ is computationally closed for $\neg\varphi$. 
%
We now show that we can repair definitions everywhere on a Herbrand model if the definitions are already satisfied on a computationally closed sub-universe. Using this result, we can repair $\Mm$ so that definitions are satisfied everywhere, which is what we want. 

\begin{lemma}[Finite Repair outside Computational Closure]
\label{lem:multi-repair-lemma}
Let $\Mm$ be a Herbrand model, $\Gamma$ a set of quantifier-free formulae, $C$ a computationally closed set for $\Gamma$, and $K$ a finite set of $\Dd$-applications not in $C$. Let $\Mm$ satisfy $\defns[C] \cup \Gamma$. Then there exists a model $\Mm'$ that satisfies $\defns[K] \cup \defns[C] \cup \Gamma$.
\end{lemma}
\begin{proof}[\textit{End of proof of Lemma~\ref{lem:multi-repair-lemma}}]
Observe that if $K$ is singleton, say $\{(D, \tup)\}$, we can construct $\Mm'$ by simply `updating' the interpretation of $D$ on $\tup$ according to the definition. Formally, the model $\Mm[D(\tup) := \llbracket\rho(\tup)\rrbracket_{\Mm}]$ satisfies $\defns[\{(D,\tup)\}] \land \defns[C] \land \Gamma$. Here $\Mm[(D,\tup):= v]$ denotes an updated model whose interpretation of $D(\tup)$ is $v$ but is otherwise identical to $\Mm$. We also use $\llbracket \cdot \rrbracket_{\Mm}$ to denote the interpretation of $\Mm$. The correctness of this construction follows from the fact that the definitions over $C$ are satisfied despite the update since $C$ is computationally closed. Consequently the satisfaction of $\Gamma$ is also unaffected because if $G(\overline{r})$ occurs in $\Gamma$ then $(G, \overline{r})$ belongs to $C$.

\smallskip
To show that $\defns[K]$ is satisfiable for an arbitrary finite subset $K$, we take $\Mm$ and apply updates as above on each pair in $K$. However, we have to do this carefully so that each repair does not break any previous repairs. Fix a set $K'$ and a model $\Mm'$ such that $K' \subseteq K$ and $\Mm'$ is $\Mm$ with updated with the fixes for the elements in $K'$. Initially $K' = K$ and $\Mm' = \Mm$. We describe below a mechanism $\mathit{Minimal}(K, K', \Mm')$ to choose a `minimal' element in $K$ that has not been fixed yet, and repair it as described above. 

\noindent
$\mathit{Minimal}(K, K', \Mm')$ is as follows:
\begin{enumerate}
    \item Pick an arbitrary element $(D, \tup) \in (K \setminus K')$. Let the body of $\defn_D$ be $\rho$.
    \item We evaluate $\rho(\tup)$ on $\Mm'$ in the following way: subterms must be evaluated before superterms, and for conditionals we evaluate the condition first and then only evaluate the appropriate branch.
    \item If the evaluation as described above does not encounter any element in $K \setminus K'$, then return $(D, \tup)$.
    \item If the evaluation of $\rho(\tup)$ encounters a term $G(\overline{r})$ such that $(G, \overline{r}) \in (K \setminus K')$, we recurse, going back to Step (2) and evaluating $\tau(\overline{r})$ where $\tau$ is the body of $\defn_G$.
\end{enumerate}

\noindent
Informally, this mechanism has the flavor of an \emph{eager evaluation}, in that we evaluate $\rho(\tup)$ eagerly, following the evaluation procedure down (recursively) to a minimal unfixed $\Dd$-application in $K$. 

\smallskip
Finally, when the procedure returns an element $(H, \overline{u})$, we add it to $K'$ and update $\Mm'$ with the repair for $(H, \overline{u})$. We then repeat this process of picking a minimal element and repairing the model on it until all elements in $K$ are fixed. This completes our construction, and the model $\Mm'$ obtained at the end of all the fixes is the model we desire.

\medskip
\rev{A subtle point is that the correctness of the above argument relies on showing that $\mathit{Minimal}(K, K', \Mm')$ terminates. However, this follows from the definition of provable acyclicity.} 
%
%
%
%
Formally: 
\begin{lemma}
\label{lem:pathcond-soundness}
Let $(D, \tup)$ be a $\Dd$-application in $K$, $\rho$ be the body of the definition of $D$, and $\Mm$ be a model of $\thcomb$. Let $(G, \overline{r})$ be a $\Dd$-application such that $G(\overline{r})$ is a sub-expression of $\rho(\tup)$ and 
the evaluation of $\rho(\tup)$ in $\Mm$ as performed in the computation of $\mathit{Minimal}$ described above encounters $G(\overline{r})$. Further, let $\psi$ be the path condition of the sub-expression $G(\overline{r})$ that is reached 
(see Definition~\ref{def:pathcond}). Then, we have that $\Mm \models \psi$.
\end{lemma}

We skip the proof of this lemma as it follows trivially from the description of the evaluation mechanism and Definition~\ref{def:pathcond}. Note that since $G(\overline{r})$ occurs in the expansion of $D(\overline{t})$, we have $\strat(D) = \strat(G)$. Now, from the definition of provable acyclicity (Definition~\ref{def:provably-acyclic}), we know:
\begin{center}
$\thcomb \models \left(\bigwedge\limits_{\strat(H) \,<\, \strat(D)}\defn_H\right) \rightarrow \psi \rightarrow \mathit{Rank}_G(\overline{r}) < \mathit{Rank}_D(\tup)$
\end{center}
Per our assumption we have only one stratum, so the set $\{H\}_{H \in \Dd, \strat(H) < \strat(D)}$ is empty. Since $\Mm$ is a $\thcomb$ model that satisfies $\psi$, we obtain $\Mm \models \mathit{Rank}_G(\overline{r}) < \mathit{Rank}_D(\tup)$. Therefore, if updating $D(\tup)$ requires updating $G(\overline{r})$, then the rank of $G(\overline{r})$ is smaller. Therefore, each recursive call of $\mathit{Minimal}$ is made on a smaller element of $K$, and therefore the evaluation of $D(\tup)$ cannot depend on itself. The mechanism for picking a minimal element is indeed well-defined and we can produce at the end of the procedure a model $\Mm'$ that satisfies $\defns[K] \cup \defns[C] \cup \Gamma$.
\end{proof}

\added{
\mypara{Repair for all tuples} We now show that we can repair definitions everywhere, i.e., on arbitrarily large sets of $\Dd$-applications. We first show some easy results. For a sort $\sigma$, consider the set $U_\sigma$ consisting of all terms of type $\sigma$. Then, the set $\DApp$ of all possible $\Dd$-applications is:

\begin{center}
$\DApp = \{ (D,(t_1,t_2,\ldots t_r)) \,\mid\, D \in \Dd, \textrm{ D has signature } \sigma_1 \times \sigma_2 \ldots \sigma_r \rightarrow \sigma,\, t_i \in U_{\sigma_i} \}$
\end{center}

Note that any $\Dd$-application $(D, \tup)$ must belong to $\DApp$, and in particular any computationally closed set $C$, which is a set of $\Dd$-applications, must be a subset of $\DApp$.

Finally, if $\Nn$ is a Herbrand model of $\thcomb$, then its universe for a sort $\sigma$ is precisely $U_\sigma$. Therefore, satisfying definitions \emph{everywhere} on $\Nn$ simply amounts to satisfying definitions on $\DApp$. The following lemma captures this idea:

\begin{lemma}[Definition Application Terms in a Herbrand Model]
\label{lem:defs-everywhere-herbrand}
Let $\Nn$ be a Herbrand model of $\thcomb$. Then, $\Nn \models \defns$ if and only if $\Nn \models \defns[\DApp]$
\end{lemma}

We skip the proof of the lemma since it is obvious. We now show the correctness of repairing arbitrarily large sets of $\Dd$-applications outside a computational closure.

\begin{lemma}[Definition Completion Lemma]
\label{thm:compclose-completion}
Let $\defns$ be a set of definitions with only one stratum. Let $C$ be a set that is computationally closed with respect to a set of quantifier-free formulas $\Gamma$. If $\defns[C] \land \Gamma$ is $\thcomb$-satisfiable, then $\defns \land \Gamma$ is $\thcomb$-satisfiable.
\end{lemma}

\begin{proof}[\textit{End of proof of Lemma~\ref{thm:compclose-completion}}]
We claim that $\defns[\DApp] \cup \Gamma$ is $\thcomb$-satisfiable. Since $C \subseteq \DApp$, let us rewrite this as $\defns[C] \cup \defns[\overline{C}] \cup \Gamma$, where $\overline{C} = \DApp \setminus C$ is the complement of $C$. 

We have from the statement of the theorem that $\defns[C] \cup \Gamma$ is satisfiable. Therefore, to show satisfiability of our desired set by compactness, it is sufficient to show that $\defns[C] \cup \Gamma \cup B$ is satisfiable for an arbitrary finite subset $B$ of $\defns[\overline{C}]$. 

Observe that a finite subset of $\defns[\overline{C}]$ is of the form $\defns[K]$ for a finite set $K \subseteq \DApp$. We are now done, since we know that $\defns[C] \cup \Gamma \cup \defns[K]$ is satisfiable from Lemma~\ref{lem:multi-repair-lemma}.

Finally, consider a model $\Nn$ such that $\Nn \models \defns[DApp] \cup \Gamma$. Without loss of generality, we can assume that $\Nn$ is a Herbrand model. Applying Lemma~\ref{lem:defs-everywhere-herbrand} gives us that $\Nn \models \defns \cup \Gamma$, which concludes the proof.
\end{proof}

\subsection*{\bfseries Epilogue: Generalizing model repair to stratified definitions}

In the above proof we assumed that $\defns$ had only one stratum. For the case of multiple strata, we first begin with the model $\Mm$ of $\infunfold$ given to us by Stage 1. We then induct on the stratum number $i$, with the inductive hypothesis being that definitions for functions from strata $< i$ are satisfied everywhere. This hypothesis is true for the base case of the lowest stratum $i=0$ by Lemma~\ref{lem:multi-repair-lemma}. 

Inductively, we assume the hypothesis and then repair the model on definitions in the current strata by applying Lemma~\ref{thm:compclose-completion}. The arguments for the correctness of repair in this case are identical, i.e., we show the correctness of finite repair and then apply compactness.

However, there is a subtlety involved in showing the correctness of finite repair. The arguments are the same as in the proof of Lemma~\ref{lem:multi-repair-lemma}, with one exception. When we consider the argument for decreasing ranks, we needed the conjunct $\left(\bigwedge\limits_{\strat(H) < \strat(D)}\defn_H\right)$ to be valid. When there is only one stratum, this is trivial since the conjunct is empty. For a general stratum $i$ in our induction proof, the formula demands that the definitions corresponding to defined functions from lower strata are satisfied everywhere. But this is precisely the induction hypothesis! This concludes the proof of correctness of repair for a stratified set of definitions.\hfill \textit{End of proof of Theorem~\ref{thm:completeness}}\hskip2pt$\blacksquare$

\bigskip
As a corollary, we obtain the following result which captures the intuition behind the definition of provable acyclicity:

\begin{corollary}[Repetition of Theorem~\ref{thm:provably-acyclic}]
\label{cor:provably-acyclic}
Given a signature $(\Ss, \Ff, \Dd, \thcomb)$, a set of stratified definitions $\defns$ that are provably acyclic, and a model $\Mm$ of $\thcomb$, there exists a model $\Mm'$ of $\thcomb$ such that the interpretation of symbols in $\Ff$ coincides with $\Mm$ and interpretations of symbols in $\Dd$ satisfy their definitions.
\end{corollary}

\begin{proof}[\textit{End of proof of Corollary~\ref{cor:provably-acyclic}}]
We simply apply the arguments of the definition completion lemma (Lemma~\ref{thm:compclose-completion}) to $\Mm$ for each stratum of the definitions starting with the lowest, setting the set $C$ where definitions are already satisfied as well as the set of formulas $\Gamma$ to be empty. We can show by induction that when we apply the lemma at stratum $i$, the resulting model will satisfy definitions at strata $\leq i$ everywhere.
\end{proof}
} 

\section{\sys reasoning in \lh}
\label{sec:lh}

Next, let us see how the \lh verifier (\slh) employs a particular
instance of \sys reasoning referred to by the tool as
\emph{reflection} and \emph{proof by logical evaluation} (PLE).
We show how a user might use \slh to develop a small library
of theorems about Peano numbers to illustrate why it can be viewed as \sys reasoning, why its \sys-style
instantiation heuristics are effective in practice, and, perhaps
more importantly, why extra information is \emph{really} required
from the user when instantiation fails.

\mypara{Peano Addition}
Consider the definition of \emph{Peano} numbers
\begin{code}
data Peano = Z | S Peano
\end{code}

As described in Section~\ref{sec:syntax-semantics},
\lh uses the above definition to generate an
ADT |Peano|  with
(1)~two \emph{constructors} $\tZ$ and $\tS$,
(2)~two \emph{recognizers} $\isZ$ and $\isS$, and
(3)~a single \emph{destructor} $\tpred$.
Next, consider the following function
that recursively defines addition
\begin{code}
  plus :: Peano -> Peano -> Peano
  plus Z     m = m
  plus (S n) m = S (plus n m)
\end{code}
\slh generates a \emph{definition} for |plus| which is an ``axiom'' constraining $\tplus$ \cite{Vazou18}
%
\begin{align}
\label{eq:def:plus}
\fdef{\tplus}\ \equiv\ \forall n, m.\ \tplus(n, m) = \tIte{\isZ(n)}{m}{\tS(\tplus(\tpred(n), m)}
\end{align}

\subsection{Proof by instantiation}

\mypara{Propositions as Types}
Suppose we wish to verify that the addition of |Z| is an identity function, \ie
the proposition
${\forall n: \tPeano.\ \tplus(\tZ, n) = n}$.
In \slh, a user uses the recipe of ``Propositions as Types''
to \emph{specify} the property as a type, and \emph{verify}
it via a function |zeroL| that inhabits the type:
\begin{code}
  {-@ zeroL :: n:Peano -> { plus Z n == n }  @-}
  zeroL n = ()
\end{code}
In the above type signature, the \emph{input parameter} has the
effect of quantifying over all $n$, and the \emph{output post-condition}
stipulates the particular property that must hold for each $n$ \cite{Wadler15}.

\mypara{Programs as Proofs}
To check this proof, \slh generates a VC
${\fdef{\tplus} \rightarrow \forall n. \tplus(\tZ, n) = n}$.
Next, it uses \emph{logical evaluation} (PLE) \cite{Vazou18} to
instantiate the definition of $\tplus$ (\ref{eq:def:plus}) 
at $(\tZ, n)$ to get the \ivc
${\forall n.\ \instat{\fdef{\tplus}}{\tZ, n} \rightarrow \tplus(\tZ, n) = n}$
using the instantiation

\centerline{
$\instat{\fdef{\tplus}}{\tZ, n} \ \equiv\ (\tplus(\tZ, n) = \tIte{\isZ(\tZ)}{n}{\tS(\tplus(\tpred(\tZ), n)}$
}

\noindent
The SMT solver proves the above \ivc is valid
even when $\tplus$ is uninterpreted, thereby
verifying that |plus Z| is an identity function.



\subsection{Proof by induction}
\label{sec:lh-induction-proofs}

\slh makes no attempt to automate inductive proofs.
Instead, the programmer must \emph{explicate induction
via recursion}, by writing programs where the induction
hypothesis is made explicit in the VC via the asserted
post-conditions of recursive calls to smaller inputs.
For example, suppose we wish to verify that the
definition of |plus| is commutative.
As before, the programmer would start by specifying
the above proposition as the following type:

\begin{code}
    {-@ comm :: n:_ -> m:_ -> {plus n m == plus m n} @-}
\end{code}

They would first attempt a \emph{direct
proof} |comm n m = ()|\footnote{$()$ is a ``unit proof'' with no extra hints from the user. \slh~ attempts to prove the VC directly given a unit proof.} which yields the \sys VC
%
\begin{equation}
\label{eq:comm:prop}
\fdef{\tplus} \rightarrow (\forall n, m.\ \tplus(n, m) = \tplus(m, n))
\end{equation}
Sadly, PLE does not find any suitable instantiations,
and so the SMT solver \emph{cannot} prove the above
is valid when $\tplus$ is uninterpreted and hence
\emph{rejects} the code on the left.

\mypara{Rogue nonstandard model}
Did \slh simply give up too early --- maybe some carefully
chosen instantiations would produce a valid \ivc?
It turns out that this is not true, and in fact we can argue this formally. 
Verification fails because formula~\ref{eq:comm:prop} 
is refuted by a rogue nonstandard model (Figure~\ref{fig:peano:non}) where the
interpretation for the constructors and destructors respects the ADT axioms for |Peano| and $\tplus$ satisfies its definition, but there exists an element $i'$ such that $\tplus(i', \tZ) \neq \tplus(\tZ, i')$ in the model. 
%
Let us banish such rogue models by proving
that adding |Z| on the \emph{right} is also an identity:
%
\begin{equation}
\label{eq:zeroR:prop}
{\forall n: \tPeano.\ \tplus(n, \tZ) = n}
\end{equation}
%

A direct proof of~\ref{eq:zeroR:prop} 
is doomed--- it yields the VC below which is refuted by 
\iffull
the rogue nonstandard model in Figure~\ref{fig:peano:non}: 
\fi
${\fdef{\tplus} \rightarrow \forall n.\ \tplus(n, \tZ) = n}$


\begin{figure}[t]
\added{
$$\begin{array}{lcllcllcl}
\multicolumn{3}{c}{\mbox{Constructor}} & \multicolumn{3}{c}{\mbox{Destructor}}  & \multicolumn{3}{c}{\mbox{Plus}}    \\
\intp{\tS}(i)  & = & i + 1    & \intp{\tpred}(i)  & = & n - 1\ \mbox{if}\ 0 < n & \intp{\tplus}(i, j)   & = & i + j \\
\intp{\tS}(i') & = & (i + 1)' & \intp{\tpred}(i') & = & (i-1)'                  & \intp{\tplus}(i', j') & = & (i + j)' \\
\intp{\tZ}     & = & 0        &                   &   &                         & \intp{\tplus}(i', j)  & = & (i + j + 1)' \\
               &   &          &                   &   &                         & \intp{\tplus}(i,  j') & = & (i + j)'
\end{array}$$
}
\caption{
\added{
Rogue Nonstandard Model for $\tPeano$ over the universe
$\univ \equiv \{ 0,1,2, \ldots \} \cup \{ \ldots,-2',-1',0',1',2',\ldots \}$.
comprising the naturals and a \emph{primed} version of each integer.
The model provides an interpretation for various constructors, destructors
and $\tplus$ that respects the ADT axioms,
but where $\tplus$ has a nonstandard interpretation on nonstandard ADT elements $i'$: $\intp{\tplus}(i', \tZ) = (i + 1)' \not = i'$, refuting~\ref{eq:zeroR:prop}, 
and
$\intp{\tplus}(i', j) \not = \intp{\tplus}(j, i')$, refuting~\ref{eq:comm:prop}. 
}
}
\label{fig:peano:non}
\end{figure}

\begin{figure*}
%
\begin{code}
-- Succeeds
{-@ zeroR :: n:_ -> {plus n Z == n} @-}
zeroR Z     = ()
zeroR (S n) = zeroR n

-- Fails
{-@ comm :: n:_ -> m:_ -> {plus n m == plus m n} @-}
comm Z     m = zeroR m
comm (S n) m = comm n m

-- Succeeds
{-@ comm :: n:_ -> m:_  -> {plus n m == plus m n} @-}
comm Z     m = zeroR m
comm (S n) m = comm n m && succR m n
\end{code}
\caption{Proof of the commutativity of $\tPeano$ addition:
The explicit case-splitting, recursion and ``lemma application''
are needed to eliminate rogue nonstandard models.}
\label{fig:zeroR}
\label{fig:comm}
\end{figure*}

\mypara{An inductive proof}
The programmer must spell out an inductive 
proof as a (recursive) piece of code that
yields a VC which \emph{excludes} rogue nonstandard models by explicitly stating the induction hypothesis as an antecedent in the VC.
This is achieved via the proof |zeroR|
shown on the left in Figure~\ref{fig:zeroR}.
First, we split cases (via a pattern match)
on the first argument, treating separately
the cases where the argument is |Z| or |S n|.
Second, the recursive call to |zeroR n| puts
the post-condition of |zeroR| for the
\emph{smaller input} |n| as a hypothesis
for the new VC
%

\centerline{
$\fdef{\tplus} \rightarrow (\forall n.\ \tplus(\tZ, \tZ) = \tZ \wedge \tplus(n, \tZ) = n \rightarrow \tplus(\tS(n), \tZ) = \tS(n))$
}

\smallskip
\noindent
PLE now instantiates $\fdef{\tplus}$ (\ref{eq:def:plus}) 
at
$(\tZ, \tZ)$ and $(\tS(n), \tZ)$ to get the \ivc
%

\centerline{
$\instat{\fdef{\tplus}}{\tZ, \tZ} \rightarrow \instat{\fdef{\tplus}}{\tS(n), \tZ} \rightarrow
                (\tplus(\tZ, \tZ) = \tZ \wedge \tplus(n, \tZ) = n \rightarrow \tplus(\tS(n), \tZ) = \tS(n))$
}
%
\added{
\smallskip
\noindent
where the instances are
\begin{align*}
\instat{\fdef{\tplus}}{\tZ, \tZ} & \equiv (\tplus(\tZ, \tZ) = \tIte{\isZ(\tZ)}{\tZ}{\tS(\tplus(\tpred(\tZ), \tZ)} \\
\instat{\fdef{\tplus}}{\tS(n), \tZ} & \equiv (\tplus(\tS(n), \tZ) = \tIte{\isZ(\tS(n))}{\tZ}{\tS(\tplus(\tpred(\tS(n)), \tZ)}
\end{align*}
} 

%

\noindent
The \ivc 
is valid even when $\tplus$ is
uninterpreted, thus proving formula~\ref{eq:zeroR:prop}. 
In essence, the (well-founded) \emph{recursive call} to |zeroR| establishes
the \emph{induction hypothesis} for the smaller |n|, thereby eliminating the rogue nonstandard models, letting us verify the proposition for \emph{any} |Peano| |n|.

\subsection{Proof by lemmas}
\label{sec:lh-lemma-proofs}

Next, let us see how to use auxiliary lemmas like |zeroR|
to eliminate rogue nonstandard models that thwarted the direct
proof of the commutativity of |plus|.
First, the programmer might attempt an inductive proof
(like the |zeroR|) as shown in the middle in Figure~\ref{fig:comm}:
split cases on whether the first parameter is |Z| or |S n|.
In the base case, they would \emph{call} |zeroR m| to
eliminate the rogue nonstandard model where 
%
\iffull
$\tplus(i', 0) \neq \tplus(0, i')$ (Figure~\ref{fig:peano:non}). 
\fi
In the inductive case, they would recursively invoke
the induction hypothesis via recursively calling |comm n m|.
This time, \slh generates the VC
%
\begin{align}
\label{eq:comm:vc:1:2}
  \fdef{\tplus} \rightarrow (\forall n, m.\ & (\tplus(m, \tZ) = m \rightarrow \tplus(\tZ, m) = \tplus(m, \tZ)) \notag \\
  \wedge\ & (\tplus(n, m) = \tplus(m, n) \rightarrow \tplus(\tS(n), m) = \tplus(m, \tS(n))))
\end{align}
%
%
%
%

Thanks to the equality asserted by the use
of the ``lemma'' |zeroR m|, the first conjunct
can proved valid via the instantiation $\instat{\tplus}{(m, \tZ)}$.
However, the second conjunct is invalid
\emph{despite} the recursive (inductive)
call to |comm n m| because of a \emph{different} rogue nonstandard model
for $\tplus$ 
that falsifies the second conjunct! 

\added{
\mypara{Rogue nonstandard model for $\mathtt{comm}$ Attempt 1}

The following is a rogue nonstandard model for the \inlinespeccode{Peano} numbers over the same ADT universe as the model in Figure~\ref{fig:peano:non} and a different interpretation of \inlinespeccode{plus} that refutes the second conjunct of formula~\ref{eq:comm:vc:1:2}. 
%
\begin{equation}
\begin{array}{lcllcl}
  \intp{\tplus}(i, j)   & = & i + j       & \quad \intp{\tplus}(i', j') & = & (i + j + 1)' \quad \mbox{if}\ 0 \leq j \\
  \intp{\tplus}(i, j')  & = & (i + j)'    & \quad \intp{\tplus}(i', j') & = & (i + j - 1)' \quad \mbox{otherwise}    \\
  \intp{\tplus}(i', j)  & = & (i + j)'    & \quad
\end{array}\notag{}
\end{equation}
The reader should take a moment to check that the
above definition respects $\fdef{\tplus}$.
However, even though the induction hypothesis trivially holds
at $\intp{\tplus}(-1', -1')$ (as the arguments are the same),
$\intp{\tplus}(0', -1') = -2'$ which differs from $\intp{\tplus}(-1', 0') = 0'$!

\mypara{Proof of $\mathtt{comm}$ Attempt 2}
The peculiar property of the rogue nonstandard
model from attempt 1 is that it introduces a \emph{discontinuity}
at $0'$ that violates 
$\forall n, m: \tPeano.\ \tplus(n, \tS(m)) = \tS(\tplus(n, m))$.
We separately specify and inductively prove this property via
a lemma \inlinespeccode{succR}:
}

\begin{code}
  succR :: n:Peano -> m:Peano -> { plus n (S m) = S (plus n m) }
\end{code}

\added{
The proof of \inlinespeccode{succR} is similar to \inlinespeccode{zeroR}:
}

\begin{code}
  {-@ succR :: n:Peano -> m:Peano -> { plus n (S m) = S (plus n m) } @-}
  succR Z     _ = ()
  succR (S n) m = succR n m
\end{code}

\added{
Now, we can use both helper lemmas \inlinespeccode{zeroR} and \inlinespeccode{succR} to
eliminate the rogue nonstandard models that thwarted
our previous attempts, by the proof shown on the right in
Figure~\ref{fig:comm}.
First, we replace the body of \inlinespeccode{comm Z m} with \inlinespeccode{comm z m = zeroR m}
which add the post-condition of \inlinespeccode{zeroR} as a lemma.
Second, we strengthen the body of \inlinespeccode{comm (S n) m} with
a call \inlinespeccode{succR n m}.
Together, these lemmas yield a VC with the strengthened antecedents
that preclude the above rogue nonstandard models
\begin{align*}
  \fdef{\tplus} \rightarrow (\forall n, m.\ & (\tplus(m, \tZ) = m \rightarrow \tplus(\tZ, m) = \tplus(m, \tZ)) \\ 
   \wedge\ & (\tplus(n, m) = \tplus(m, n) \rightarrow \tplus(m, \tS(n)) = \tS(\tplus(m, n)) \rightarrow \\ 
    & \quad \tplus(\tS(n), m) = \tplus(m, \tS(n)))) 
\end{align*}
This time, PLE instantiates $\fdef{\tplus}$ at $(\tZ, m)$ and $(\tS(n), m)$
to yield an \ivc that the SMT solver validates even when $\tplus$ is uninterpreted.
Note that as with induction, \slh makes no attempt to automate the creation
and use of such lemmas: the programmer must explicitly spell them out by defining
and proving them, and then ``calling'' the lemmas inside the theorem body to
appropriately ``instantiate'' them at the relevant values, thereby yielding a VC that can be automatically discharged by \sys reasoning.
} 

\subsection{Rogue nonstandard models in proofs about data structures}
\label{sec:intuitive-ns-model}

We use the simple |Peano| datatype to illustrate how \slh implements \sys reasoning, and how direct proofs can fail due to rogue nonstandard models which can be eliminated via explicit induction (recursion) and lemmas (function calls).
Similar phenomena occur when verifying more complicated
properties. Consider the datatype of finite
maps from keys (|k|) to values (|v|)
\begin{code}
  data Map k v = Leaf | Node k v (Map k v) (Map k v)
\end{code}
Figure~\ref{fig:bst} shows the code for two functions that
respectively |get| the value of a |key| from a tree, and
|set| the value of a |key| to some new |val| leaving the values
of all other keys unchanged.
\begin{figure*}[t]
\begin{minipage}[t]{.42\textwidth}
\begin{code}
get :: Map k v -> k -> Maybe v
get (Node k v l r) key
 | key == k  = Just v
 | key <  k  = get l key
 | otherwise = get r key
get Leaf _   = Nothing
\end{code}
\end{minipage}
\begin{minipage}[t]{.57\textwidth}
\begin{code}
set :: Map k v -> k -> v -> Map k v
set (Node k v l r) key val
 | key == k  = Node key val l r
 | key <  k  = Node k v (set l key val) r
 | otherwise = Node k v l (set r key val)
set Leaf k v = Node k v Leaf Leaf
\end{code}
\end{minipage}
\caption{Implementations of $\mathtt{get}$ and $\mathtt{set}$ functions for Binary Search Tree.}
\label{fig:bst}
\end{figure*}
The following proposition is one of McCarthy's two laws that
characterize finite maps

\centerline{$\forall m, k, v.\ \tget(\tset(m, k, v), k) = \tJust(v)$}
%

\added{
\mypara{Attempt 1: Direct proof}
In \slh, we could try to specify and
verify the above law as
}
\begin{code}
  getEq :: m:_ -> k:_ -> v:_ -> { get (set m k v) k = Just v }
  getEq m k v = ()
\end{code}
\added{
which would yield the VC
${ ( \fdef{\tget} \wedge \fdef{\tset} )  \rightarrow \forall m,k,v.\ \tget(\tset(m, k, v), k) = \tJust(v)}$
Unfortunately, no (finite) instantiation can prove the above VC, and we describe here an intuitive rogue nonstandard model that falsifies the theorem.
}

\textit{Rogue Nonstandard Model~}
Let the universe $\univ$ be the set of \emph{finite trees} and \emph{infinite non-regular trees}\footnote{A non-regular tree is one that is not isomorphic to any of its proper subtrees. This is a technical condition we require to ensure that the ADTs are \emph{acyclic}, i.e., it is not possible to reach a term by destructing itself.} over $(k, v)$ pairs.
The interpretation for $\tget(m, k)$ on a tree $m$
follows the usual path in a binary search tree to find $k$, even on infinite trees.
If the $k$ is found, $\tget$ 
returns the corresponding value, and if the path ends or continues forever, then $\tget$
returns |Nothing|.
%
%
%
%
%
The interpretation for $\tset$ is similar,
except that on an infinite computation it
returns the input tree.

\rev{To see why this model refutes the VC, consider an infinite non-regular binary tree $m$ that does not contain $1$ as a key. For instance, the nodes of $m$ can have keys labeled according to a level-order traversal, starting from $2$.}
Let us call |set|, to set key $1$ to the value $1$
and try to $\tget$ the value of key $1$ after
that, \ie consider $\tget(\tset(m, 1, 1), 1)$
By the above interpretation $\tset(m, 1, 1) = m$,
as all paths in $m$ are infinite, and further
$\tget(m, 1) = \mathtt{Nothing}$ thereby refuting
the proposition despite being a model of the ADT
theory and the definitions of |get| and |set|. Intuitively, |set| loses the update entirely, and hence |get| returns |Nothing|.

\added{
\mypara{Attempt 2: Inductive proof}
We cannot prove the law directly. Instead, we need an inductive proof as in Figure~\ref{fig:zeroR}.
}
\begin{code}
  getEq (Node key _ l r) k v
    | k == key   = ()
    | k <  key   = getEq l k v
    | otherwise  = getEq r k v
  getEq Leaf _ _ = ()
\end{code}
\added{
The successful proof splits cases on constructor
for \inlinespeccode{m} and then on the ordering of the two keys,
and recursively invokes \inlinespeccode{getEq} (\ie applies the
induction hypothesis) on the left and right subtrees
appropriately.
The induction hypotheses in the antecedents of the VC for the above 
eliminates the rogue nonstandard models, and allows PLE to find suitable instantiations such that the resulting \ivc can be validated by the SMT solver.
}

\added{
\subsection{Comparison between PLE~\texorpdfstring{(\cite{Vazou18})}{} and \qialg}
\label{sec:lh-ple}

Although our development of the \qialg\ algorithm is inspired from the PLE (Proof by Logical Evaluation) heuristic introduced for \liquid in~\cite{Vazou18}, there are some differences between the two algorithms. The work in~\cite{Vazou18} shows that PLE is in fact complete if an `equational proof' exists. However, this result is much weaker than ours, and in fact PLE fails to prove simple theorems that \qialg\ can prove. We illustrate this difference using the following example:
}

\begin{code}
  g :: Int -> Int
  g x = 1

  h :: Int -> Int
  h x = 2

  f :: Int -> Int
  f x
    | x > 0   = g x
    | otherwise  = h x
\end{code}

\added{
\smallskip
Let us consider the property $\forall x.\, f(x) > 0$. The PLE heuristic implemented in \slh cannot prove this property. Intuitively, in the case of this example PLE would only unroll a definition if it can show that the terms used for instantiation satisfy the branch conditions, effectively only unrolling the definition for one case. Therefore, it cannot unroll either case in the definition of $f$ for an arbitrary $x$ and ends up not using any instantiations. To fix this, the user would have to explicitly provide a hint by splitting cases. 

In contrast, it is easy to see that \qialg~ will succeed as unfolding definitions on $x$ yields 

\centerline{$(f(x) = \ite(x > 0, g(x), h(x)) \land g(x) = 1 \land h(x) = 2) \rightarrow (f(x) > 0)$}

\noindent
which is FO-valid. \qialg\ is complete with respect to the combined first-order theories, but may perform a greater number of instantiations than PLE would in general. We discuss some related work on thriftiness and completeness of instantiations in Section~\ref{sec:relatedwork}.
} 
\section{\sys reasoning and reasoning in \leon}
\label{sec:leon}

The \leon system and its successor \stainless~\cite{stainless} reason with functional programs~\cite{SuterKoksalKuncak,BlancKuncakKneussSuter,suter10} using techniques broadly similar to \slh and \qialg. \leon reasons about Scala programs with quantifier-free pre/post conditions, and the recursively defined functions occurring in annotations are written in Scala as terminating functions. It also automates certain induction proofs, 
including  induction by ``stack-height'' (akin to Hoare-style reasoning), as well as structural induction on ADTs. \leon caters to other aspects of development as well, including techniques similar to bounded model-checking 
for finding errors. We do not discuss these aspects here as they are not relevant to our work.

The first observation is that verification conditions for \leon programs can also be modeled as $(\defns, \varphi)$ in the \sys fragment. Specifically, the property $\varphi$ is quantifier-free (implicitly universally quantified) and definitions are proven terminating.

While reasoning in \leon also involves unfolding definitions followed by SMT solving, there are important differences in comparison to \slh or \qialg. First, whereas \slh typically unfolds definitions only once, \leon continually unfolds definitions over multiple rounds similar to \qialg. Second, \leon asserts the contract of a function along with its definition on the given input arguments during the unfolding. In theory, it does so for \emph{every} unfolding, \emph{ad infinitum}. 
Observe that when expressing a problem as $(\defns, \varphi)$, contracts cannot be assumed for all subsequent function calls, as that would require universally quantified assumptions. 
Assuming contracts for subsequent function calls is also strictly more powerful than simply unfolding definitions, as we illustrated in Section~\ref{sec:insert-example-contracts}.

\begin{algorithm}
\begin{algorithmic}[1]
\Statex {\bf Input:} $\defns = \{\defn_D \mid D \in \Dd\}$ and $\mathit{Contracts} = \{\psi_D \mid D \in \Dd\}$, where contracts are universally quantified, as well as the universally quantified property $\varphi$ to be proven
\Statex {\bf Output:} $\mathit{VALID}$ (when it terminates)
\Statex {\bf Imports:} $\smtproc$ for deciding $\thstd$-satisfiability of quantifier-free formulas
\Procedure{\textsc{{\bfseries \qialg-Contracts}}$[\Ss; \Ff; \Dd; \thstd]$}{}
\State $\mathit{unfoldings} := \{\neg \varphi\}$ \CodeComment{Negate and Skolemize}
\State $\mathit{contracts} := \{\}$
\While{$\mathit{True}$}
\label{leonalg:outer-loop}
\State $\mathit{res} := \smtproc(\mathit{unfoldings} \cup \mathit{contracts})$ \CodeComment{Prove $\varphi$ with unfoldings + contracts}\!
\label{leonalg:checksat}
\If{$\mathit{res} = \mathit{UNSAT}$} 
\State \Return{$\mathit{VALID}$} \CodeComment{Success}
\Else
\State \CodeComment[\linewidth]{Compute $\Dd$-applications occurring in $\mathit{unfoldings}$}
\State $\unfolded := \{ (D,\tup) \mid D(\tup) \textrm{ occurs in } \psi \textrm{ for } \psi \in \mathit{unfoldings}\}$
\label{leonalg:get-terms}
\State \CodeComment[\linewidth]{Unfold the definitions and assert the corresponding contracts}
\State $\mathit{unfoldings} := \mathit{unfoldings} \cup \{ \defn_D(\overline{t}) \mid (D,\overline{t}) \in \unfolded\}$
\label{leonalg:add-instantiation}
\State $\mathit{contracts} := \mathit{contracts} \cup \{  \psi_D(\overline{t})\mid (D,\overline{t}) \in \unfolded\}$
\label{leonalg:add-contracts}
\EndIf
\EndWhile
\EndProcedure
\end{algorithmic}
\caption{\qialg\ (Algorithm~\ref{fig:qialg}) adapted to \leon-style reasoning. The additional step compared to Algorithm~\ref{fig:qialg} is line~\ref{leonalg:add-contracts}. When there are both pre- and post-conditions given, the contract considered to be pre $\rightarrow$ post.}
\label{fig:leonalg}
\end{algorithm}

\rev{
\mypara{Reduction to \qialg\ and completeness for \leon} Algorithm~\ref{fig:leonalg} shows the \qialg\ algorithm adapted to \leon-style reasoning. Given a set of definitions along with their contracts, the \leon-style algorithm unfolds definitions similar to \qialg, and additionally asserts their contracts (line~\ref{leonalg:add-contracts}). As with \qialg\, Algorithm~\ref{fig:leonalg} also unfolds definitions and asserts contracts until the goal is proven, or continues indefinitely. 

We show that the reasoning mechanism in \leon can in fact be captured in the \sys framework and proven using \qialg. Formally, given a set of definitions $\defns$, their contracts, and the property $\varphi$ to be proven, we construct a pair $(\defns',\varphi')$ such that running \qialg\ on $(\defns', \varphi')$ is equivalent to running Algorithm~\ref{fig:leonalg} on $\varphi$. The following theorem captures this:

\begin{theorem}
\label{thm:leonalg-to-qialg}
Given a set of provably acyclic definitions $\defns = \{\defn_D \mid D \in \Dd\}$, their contracts $\{\psi_D \mid D \in \Dd\}$ with the same signature, and a universally quantified property $\varphi$, there exists an effectively computable set of provably acyclic definitions $\defns'$ and a universally quantified formula $\varphi'$ such that Algorithm~\ref{fig:qialg} applied to $(\defns', \varphi')$ simulates the verification of $\varphi$ using Algorithm~\ref{fig:leonalg}: the two algorithms either both return $\mathit{VALID}$ or both diverge, and across a diverging run the limit value (i.e., accumulated across the diverging run) of the set $\mathit{unfoldings} \cup \mathit{contracts}$ in the run of \qialg-Contracts (Algorithm~\ref{fig:leonalg}) is the same as the limit value of the set of instantiated formulas (``$\mathit{formulas}$'') in the corresponding run of \qialg (Algorithm~\ref{fig:qialg}).
\end{theorem}

\begin{proof}
We sketch the core of the proof here. The key idea is to construct, for every $G \in \Dd$, a recursively defined predicate $\mathit{Contract}_G$ with the same input signature that returns a Boolean value that computes whether the contract of $G$, $\psi_G(\overline{x})$, holds for the called parameters $\overline{x}$, as well as all the contracts $\psi_D(\overline{y})$ corresponding to the recursive calls $D(\overline{y})$ that are encountered in the computation of $G(\overline{x})$. In other words, asserting $\mathit{Contract}_G(\overline{x})$ asserts all the contracts that would be encountered in an unfolding of $G(\overline{X})$, as well as the contract for $G(\overline{x})$ itself. We then write a new verification condition that augments the property to be verified with extra assumptions that assume the contract predicates for the first level of unfolding. \qialg\ operating on this VC will then naturally assume the contracts for successive unfoldings, mimicking \leon-style reasoning. Theorem~\ref{thm:completeness} then yields completeness with respect to the underlying combined theory.


We present the detailed construction for a simplified scenario for ease of exposition. Let us fix an ADT universe with two destructors $d_1$ and $d_2$ that map to the ADT sort and other destructors that map to background sorts, and a single nullary constructor $\nil$. We also assume that there is only one recursively defined function $G$ that takes arguments $(x,\overline{y})$ where $x$ is of an ADT sort. Without loss of generality, let us also suppose that the definition of $G$ on arguments $(x,\overline{y})$ has two recursive calls with arguments $(d_1(x), \overline{t_1}(x,\overline{y}))$ and $(d_2(x), \overline{t_2}(x,\overline{y}))$ where $\overline{t_1}$ and $\overline{t_2}$ are tuples of terms over $x,\overline{y}$. Further, let the path condition to the recursive calls simply be the guard $x \neq \nil$. This definition is in the \sys fragment, i.e., provably acyclic, since the first parameter decreases in the subterm ordering on the ADT sort. Finally, let $\psi_G(x,\overline{y}, \rho)$ be the contract for $G$, with $x$ of type ADT and $\overline{y}$ from background sorts as the input parameters and a variable $\rho$ of the appropriate sort denoting the return value of $G$ on $(x,\overline{y})$.

Let $\varphi$ be the property that we want to prove. Let us assume that $G$ occurs in $\varphi$ only once, as the term $G(z,\overline{w})$. To model \qialg-Contracts, we need to be able to assume that $\psi_G(x,\overline{y}, f(x,\overline{y}))$ holds on arguments obtained by unfolding the definition of $G$ on $(z,\overline{w})$ arbitrarily many times.
 
Let us define a new recursive predicate $\cont_G(x, \overline{y})$ with input parameters identical to $G$, along with the following definition:
\begin{align*}
\forall x.~~ \cont_G(x,\overline{y}) = &\left( \psi_G(x, \overline{y}, f(x,\overline{y}))   \,\wedge\right.\notag\\
&\left.\left( x \neq \nil \rightarrow (\cont_G(d_1(x),\overline{t_1}(x,\overline{y})) \wedge \cont_G(d_2(x)),\overline{t_2}(x,\overline{y}))\right) \right)
\end{align*}

\noindent
The above definition says $\cont_G(x,\overline{y})$ is true iff $G$ satisfies its contract on the input $x,\overline{y}$ and further, if $x$ is not $\nil$, $\cont_G$ also holds for the recursive calls $(d_1(x), \overline{t_1}(x,\overline{y}))$ and $(d_2(x), \overline{t_2}(x,\overline{y}))$ that occur in the definition of $G$. Therefore, asserting $\cont_G(x, \overline{y})$ can be seen as asserting that the contract of $G$ holds for $(x,\overline{y})$ as well as all the tuples that occur when unfolding the definition of $G$ on $(x,\overline{y})$ \emph{ad infinitum}, i.e., the computational closure of $(x, \overline{y})$ (see Section~\ref{sec:completeness}). Note that we respect the path condition for the recursive calls of $\cont_G$, `aligning' them with the structure of the computation of the original definition. The path condition is $x \neq \nil$ in this case per our simplifying assumption, but one can analogously construct $\cont$ in the general case as well. Since $G$ is provably acyclic, $\cont_G$ is as well. We set $\defns'$ to be $\defns \cup \{\defn_{\cont_G}\}$.

Finally, to simulate \leon-style reasoning we want to assert the contract predicate for the topmost occurrence of $G$ in $\varphi$:

\begin{center}
$\varphi' \equiv \cont_G(z,\overline{w}) \rightarrow \varphi$
\end{center}

\noindent
where $G(z,\overline{w})$ is the only occurrence of $G$ in $\varphi$. If there are more occurrences, we similarly conjoin the corresponding $\cont_G$ calls in the antecedent.



We provide a concrete example to  illustrate the above construction. Consider definition of the usual $\mathit{length}$ function on lists, along with the contract $\mathit{length}(x) > 0$. The corresponding $\cont$ predicate is as follows:
\begin{align*}
\cont_\mathit{length}(x) = (\mathit{length}(x) > 0 \land (x \neq \nil \rightarrow \cont_\mathit{length}(\tail(x))))
\end{align*}

Now, observe that \qialg\ applied to a VC of the form $\cont_\mathit{length}(x) \rightarrow \varphi$ would, across rounds, add $\mathit{length}(x) > 0$, $\mathit{length}(\tail(x)) > 0$, $\mathit{length}(\tail(\tail(x))) > 0$ and so on to the quantifier-free query (along with definition unfoldings), mimicking \leon-style reasoning as desired.

We show the theorem by induction, arguing that in each round the two algorithms essentially instantiate the same set of formulas. The inductive hypothesis is that in each round the value of the expression $\mathit{unfoldings} \cup \mathit{contracts}$ in the run of \qialg-Contracts with $\defns$, $\mathit{Contracts}$, and $\varphi$ is the same as the value of $\mathit{formulas}$ in the corresponding round in the run of \qialg on $(\defns',\varphi')$. We skip detailing this argument since it essentially follows from an inspection of Algorithms~\ref{fig:qialg} and~\ref{fig:leonalg}.


We hence reduce \leon-style reasoning to \qialg. It is easy to see that the above construction can be generalized beyond the simplifying assumptions we made, including multiple ADT sorts with different signatures as well as multiple mutually recursively defined functions with their respective contracts.
\end{proof}
}

\smallskip
As far as we know the above result is new. Prior literature on \leon~\cite{suter10,BlancKuncakKneussSuter,SuterKoksalKuncak} shows soundness of the procedure. Restricted fragments~\cite{suter10} involving certain kinds of ``measures'' (functions from ADTs to background sorts) have been shown to admit complete unfolding based reasoning with respect to the \emph{standard model}, with a \emph{decidable} validity problem. In contrast, we show completeness (i.e., recursively enumerable procedures) for validity with respect to the \emph{combined theory} for a more general class of functions. Further, our logic is \emph{undecidable} (see Section~\ref{sec:boundarythms}), which shows that it is fundamentally different from decidable subclasses reported in prior art~\cite{suter10} (see also Section~\ref{sec:relatedwork}). 

Our results also show that when theorems are not provable in \leon, there ought to be rogue nonstandard models. We considered a few such examples  
and were indeed able to construct rogue nonstandard models. We describe one such example below.

\added{

\subsection{Rogue nonstandard models for \leon-style VCs}
\label{sec:leon-nonstd-model}

Consider the usual ADT of lists over integers, and the function $\mathit{rev}$ that computes the reverse of a list as follows:

$\forall x: \mathsf{List}.\, \reverse(x) = \;\ite(\isnil(x),\mathit{True}, \extend(\reverse(\tail(x)), \head(x)))$

\noindent where

$\forall x: \mathsf{List}, k: \mathsf{Int}.\, \extend(x, k) = \;\ite(\isnil(x), \cons(k, \nil), \cons(\head(x),\extend(\tail(x), k)))$

Consider the property $\forall x: \mathsf{List}.\, \reverse(\reverse(x)) = x$. We find that the \leon tool fails to prove this automatically. It follows then from our theorem on the completeness of \leon-style unfolding that there exists a rogue nonstandard model on which the VC generated by \leon does not hold. More precisely, since \leon attempts multiple rounds of unfolding along with assuming contracts, there exists a rogue nonstandard model on which none of the generated VCs hold. 

\mypara{Unfolding once} To build intuition, we first illustrate a rogue nonstandard model for the first level VC obtained from \leon-style unfolding with contracts. We encourage the reader to assure themselves that this VC is essentially
\begin{center}
$\big(\defn_\reverse \land \defn_\extend\big) \rightarrow \big(\reverse(\reverse(\tail(x))) = \tail(x) \rightarrow \reverse(\reverse(x)) = x \big)$    
\end{center}

A rogue nonstandard model that falsifies the above formula would satisfy the definitions of $\reverse$ and $\extend$, but would contain an element $x$ for such that applying $\reverse$ twice on $\tail(x)$ yields itself, but applying it twice on $x$ does not. We construct this model using the same universe considered in Section~\ref{sec:sorted-list-membership}. Namely, we consider the universe consisting of all finite sequences of integers (which correspond to the standard ADT Lists), as well as nonstandard elements of the form $(s, i)$ where $s$ is an infinite sequence of integers and $i$ is an integer ``index''. $\cons$ behaves as expected on standard elements, and on nonstandard elements prepends to the sequence component and increments the index component. $\head$ and $\tail$ can be obtained as the appropriate inverses of this definition. 

In this model, we define $\extend$ and $\reverse$ on standard elements in the usual way, but on the nonstandard elements we do the following: $\extend$ simply loses the given value, i.e., returns the input list, and $\reverse$ maps every nonstandard element to a default nonstandard element $d$, say $([1,1,1,\ldots],0)$. It is easy to see that these interpretations satisfy the recursive definitions above, and we encourage the reader to follow the definitions carefully and check that this is indeed the case.

Now, consider $x = ([2,1,1,1,\ldots], 1)$. This is in fact the element $\cons(2,d)$ where $d$ is the default value in the interpretation of reverse defined above. Following the interpretations, one can see that $\reverse(\reverse(x)) = d$ and $d \neq x$, but $\reverse(\reverse(\tail(x))) = d = \tail(x)$, which violates the \leon-style VC for the first level of unfolding.

\mypara{Unfolding ad infinitum} It turns out that we can similarly construct a rogue nonstandard model for all the VCs obtained from unfolding ad infinitum. Formally, we consider the set of formulas:

\begin{center}
$\big(\defn_\reverse \land \defn_\extend\big) \rightarrow \big(\left(\bigwedge_{i=1}^{n}\reverse(\reverse(\tail^i(x))) = \tail^i(x)\right) \rightarrow \reverse(\reverse(x)) = x \big)$\\
\hfill for every $n \in \mathbb{N}, n \geq 1$
\end{center}

To falsify the $i$-th formula above, a model would have to satisfy the definitions of $\extend$ and $\reverse$, and contain an element $x$ such that applying $\reverse$ twice behaves like the identity on all the $j$-th tails of $x$ for $1 \leq j \leq i$, but this property does not hold on $x$ itself.

We construct this rogue nonstandard model using a universe consisting of finite sequences of integers corresponding to the standard Lists, and a set of nonstandard elements of the form $((\mathit{front}, i),(\mathit{back},j))$ where $\mathit{front}$ and $\mathit{back}$ are infinite sequences of integers and $i,j$ are integers. Intuitively, $\mathit{front}$ represents the sequence of elements obtained by reading the nonstandard list from the front, and $\mathit{back}$ represents the sequence of elements obtained by reading the list in reverse starting from the end. $i$ and $j$ are indices corresponding to the sequences similar to the rogue model described above. However, as opposed to the model above where nonstandard elements had an infinite ``tail'', nonstandard elements in this model have an infinite ``middle''. $\cons$ behaves as expected on standard elements, and on nonstandard elements behaves similarly to the earlier rogue model: $\cons(k, ((s_1,i),(s_2,j))) = ((k :: s_1,i+1),(s_2,j))$. $\head$ and $\tail$ can then be derived from this interpretation. These interpretations satisfy the usual first order axiomatization of ADTs, and therefore is a faithful model of the theory of ADTs.

The interpretation of $\extend$ is the intuitive one. Just as $\cons$ added the value to the front, $\extend$ adds the value to the back: $\extend(((s_1,i),(s_2, j)), k) = ((s_1,i),(k :: s_2, j+1))$.

The interpretation of $\reverse$ is more involved. We interpret it as expected on standard elements. On nonstandard elements, we define $\reverse((\mathit{front},\mathit{back})) = (\mathit{back}, \mathit{front})$ for elements where $\mathit{back}$ does not begin with a $0$, otherwise we define it to be $(\cons(0,\mathit{back}), \mathit{front})$. Informally, we swap the front and back sequences to reverse the element (as is intuitive), but if the back begins with a $0$ we duplicate the value during reversal. Note that for nonstandard elements where neither the front nor the back sequence begin with a $0$, applying $\reverse$ twice yields the original element.

As earlier, we encourage the reader to follow the recursive definitions given above and check that these interpretations indeed satisfy the definitions. In particular, the recursive definition of $\reverse(x)$ extends the reverse of $\tail(x)$ with $\head(x)$, which allows us to add a spurious value to the front of the reversed tail since $\extend$ only adds to the back.

Finally, we consider $x = (([0,1,2,3,\ldots],0),([1,2,3,\ldots],0))$. Observe that all elements of the form $\tail^i(x)$ for $i \geq 1$ are elements $([i,i+1,i+2,\ldots],-i),([1,2,3,\ldots],0))$, where neither the front nor the back sequence begins with a $0$. Therefore applying $\reverse$ twice yields the original element as desired, but this does not hold for $x$ since the back of $\reverse(x)$ would begin with a 0, which adds an extra $0$ when applying $\reverse$ a second time.

\mypara{Remarks} The above models are nontrivial constructions in several ways. First, there appears to be no clear recipe or taxonomy of nonstandard universes to consider for constructing rogue models. \rev{There is some recent preliminary work in this direction~\cite{EladSeplogicCompleteness2026} for rogue models arising in entailments over Separation Logic with inductive definitions, but it largely identifies only structures with infinite `tails'.} We were also mostly able to find rogue models by using lists with infinite tails in our anecdotal study, but this seemed to be insufficient for the second set of VCs considered above where we instead resorted to lists with infinite middles. We do not know of a rogue model for the second case using lists with infinite tails only, and we also do not know how to prove that such a model cannot exist\footnote{Such claims are especially hard to make as one needs to consider the logical fragment in which interpretations can be defined}. Additionally, constructing interpretations of theory-specific symbols (like $\cons$) so that the interpretations collectively respect an ADT axiomatization seems challenging. Second, although the above interpretations are broadly intuitive, critical portions of the construction such as the definition of $\reverse$ in the second model above require careful consideration in order to simultaneously satisfy the recursive definitions and falsify the property of interest. Finally, there is an additional sense in which the second construction above is nontrivial: to show that \leon cannot prove the property no matter how many unfoldings (with contracts) it uses, i.e., to show that the proof of the property is truly beyond the methodology that \leon uses, the model must simultaneously falsify the infinite set of formulas corresponding to all the levels of unfolding.

\smallskip
In practice, the above nonstandard models are eliminated by an appropriate inductive lemma provided by the user. 

}  

\section{Expressiveness Results on the \sys Fragment}
\label{sec:boundarythms}


We show some technical results pertaining to the \sys fragment. 
%

\mypara{Undecidability of the \sys fragment} 

We show undecidability of the \sys fragment even when the combined theory admits decision procedures for quantifier-free reasoning. In other words, the validity problem for the \sys fragment,  for which we proved UQFR is a semi-decision procedure in Section~\ref{sec:completeness}, does not admit any decision procedures.

\begin{theorem}
\label{thm:boundary-undec}
The validity problem for \sys formulas is undecidable. 
\end{theorem}



\added{
We provide a reduction from the non-halting problem for two-counter machines. A two-counter machine~\cite{HopcroftMotwaniUllman} is a machine with two registers that can contain unbounded integers. The machine can only increment or decrement these counters, or check whether they are equal to zero. Two-counter machines are computationally equivalent to Turing machines~\cite{HopcroftMotwaniUllman}, and checking the halting/non-halting of a two-counter machine is undecidable (assuming, without loss of generality, an initial configuration where counters are set to zero).

It is tempting to try to find a simple reduction that encodes executions of the machine using ADTs (say, as lists of configurations), defining a recursive predicate that identifies \emph{halting} executions (which are finite), and stating the theorem that no ADT element encodes a halting execution of the machine. However, note that we are seeking validity with respect to the \emph{combined theory} and not validity in the standard model. In fact, since validity over the combined theory is recursively enumerable, we cannot reduce non-halting problem of two counter machines (which is co-r.e. hard) to it. Our reduction reduces the non-halting problem to the complement of validity, \ie, satisfiability. We provide the proof below. 
}

\begin{proof}
\added{
Fix a two-counter machine $M$.
Let us consider ADTs that are lists of triples of integers:
ADT $\mathsf{List}$, with two constructors $\mathsf{Nil}$ and $\mathsf{Cons}$:
}

\begin{code}
data List= Nil 
           | Cons (state: Int) (fst: Int) (snd: Int) (tail: List)
\end{code}
%

\added{
Each element of the list represents a configuration of a two-counter machine-- the state of the two-counter machine and the value of the two counters. We can write quantifier-free logical formulae $\mathsf{init}(x)$ representing the initial configuration, $\mathsf{halt}(x)$ representing any halting configuration, and $\mathsf{nextconfig}(x, y)$ representing that $y$ is the successor configuration of $x$. Now consider the following recursive definition $\defn_\mathsf{nonhalt}$:
\begin{align*}
    \forall x:\mathsf{List}.\,\mathsf{nonhalt}(x) =\; &\ite(x = \nil, \mathit{False}, \\
    &\;\;\ite(\mathsf{halt}(x), \mathit{False}, \mathsf{nextconfig}(x, \tail(x)) \land \mathsf{nonhalt}(\tail(x))
\end{align*}


Note that this function is provably acyclic since it recurses on $\tail(x)$. 

Consider the property $\varphi \equiv \forall x.  (\mathsf{init}(x) \rightarrow \neg~\mathsf{nonhalt}(x))$. We claim $\defn_\mathsf{nonhalt} \rightarrow \varphi$ is valid in the combined theory if and only if the two-counter machine \emph{halts}.

If the formula is not valid, then there is a model
and an ADT element $x$ such that $\mathsf{init}(x)$ and $\mathsf{nonhalt}(x)$ hold. Then there are two cases: $x$ corresponds to a finite list (reaching $Nil$ in finitely many destructions) or it is a nonstandard element corresponding to an infinite list. The former case is impossible, as no finite list can have $\mathsf{nonhalt}$ to be true on it. In the second case,
the recursive definition of $\mathsf{nonhalt}$ ensures
that the list pointed to by $x$ encodes an execution of the two-counter machine, and hence the machine does not halt.

Conversely, assume the machine does not halt. It turns out that we can build a nonstandard model where $x$ points to a nonstandard ADT element encoding an infinite list that corresponds to the non-halting execution of the machine. Formally, we use the compactness theorem instead of constructing this model explicitly. Note that for any $k \in \mathbb{N}$, there is a \emph{standard} ADT element that encodes a finite list corresponding to a partial execution of the two-counter machine for $k$ steps. Hence any \emph{unfolding} of the definition of $\mathsf{nonhalt}$ on $x$, $tail(x)$, etc.
up to $k$ destructions is satisfiable. By the compactness theorem, the unfolding for $\omega$ number of steps is also satisfiable. $\mathsf{nonhalt}(x)$ will hold in this model.
}
\end{proof}

\mypara{Incompleteness with terminating definitions} 

It is natural to ask whether \qialg\ is complete for all terminating functions, not just provably acyclic ones. We show that this is not the case.

\begin{theorem}
\label{thm:boundary-incompleteness}
There exists a signature $(\Ss, \Ff, \Dd, \thcomb)$, a set $\defns$ of well-defined definitions for $\Dd$ that are not provably acyclic, and a universally quantified formula $\varphi$ such that $\thcomb \models (\defns, \varphi)$ but $\qialg$ does not terminate.
\end{theorem}
Note that this theorem implies that generalizing definitions to arbitrary universally quantified formulas also leads to incompleteness.

\begin{proof}
\added{
We construct an instance of the validity problem without provably acyclic definitions that is unprovable for \qialg. We use the ADT of lists over integers as our foreground sort:
}

\begin{code}
    data List = Nil | Cons (head : Int) (tail : List)
\end{code}

\added{
\noindent
as well as the following definitions:
\begin{align*}
\defn_{\mathit{std}}\; \equiv&\; \forall x:\mathsf{List}.\,\mathit{std}(x) =\; \ite(x = \nil, \mathit{True}, \mathit{std}(\tail(x))) \\
\defn_{R}\; \equiv&\; \forall x:\mathsf{List}.\, R(x) =\; \ite(\mathit{std}(x), \mathit{True}, \ite(\head(x) = 0, R(x), \neg R(x)))
\end{align*}

Both functions are well-defined definitions as they terminate on the standard model. The termination of $\mathit{std}$ is apparent: it simply destructs the input term recursively until $\nil$ and then returns $\mathit{True}$. $R$ is also terminating on the standard model since $\mathit{std}(x)$ is always true on standard elements, therefore the \emph{else} branch of the outer $\ite$ is never taken.

$\mathit{std}$ is also provably acyclic since the arguments to the recursive call are smaller according to the subterm ordering. However, $R$ is not provably acyclic. We demonstrate this indirectly by proving the incompleteness of \qialg\ and defer the discussion of why it is not provably acyclic.

Consider the theorem $\varphi \equiv \forall y:\mathsf{List}.\, R(y)$. We claim: (1) $\thcomb \models \defns \rightarrow \varphi$, and (2) \qialg\ does not terminate on $(\defns, \varphi)$. We do not prove the latter here as it can be deduced easily by following Algorithm~\ref{fig:qialg} and only show the former.

Suppose that the claim is not true. Note that the antecedent is not vacuous since it is possible to satisfy the definitions on the standard model. Therefore, for the claim to be false there must exist a model where the definitions are satisfied, but there is a $y$ such that $R(y)$ does not hold. From the definition of $R$, we know that this is only possible when $\neg\mathit{std}(x)$ and $\head(x) = 0$. Any other path in the definition either leads to $R(x)$ being true or the impossibility $R(x) = \neg R(x)$. Now, consider the element $\cons(1,y)$. From the definition of $\mathit{std}$, we have $\mathit{std}(\cons(1,y)) = \mathit{std}(y) = \mathit{False}$. Then, following the definition of $R$ yields $R(\cons(1,y)) = \neg R(\cons(1,y))$, which is impossible. Therefore, it must be the case that there is no model where the definitions are satisfied but $\varphi$ does not hold. In other words, $\thcomb \models (\defns, \varphi)$.

\qialg\ never terminates on the algorithm because unfolding the definitions only ever produce terms that are destructions of $y$, whereas we proved the validity above by instantiating the definitions on a superterm of $y$. This shows that \qialg\ is incomplete for this instance.
}
\end{proof}

\added{
Since we prove completeness for provably acyclic definitions, the above shows that $R$ is not provably acyclic. More specifically, this is because in order to prove that the absurd recursive call $\neg R(x)$ is unreachable, we must essentially prove that $\mathit{std}(x)$ always holds. However, this not true in the combined theory, and one would typically use induction to establish this.

Consequently, in models where $\mathit{std}(y)$ does not hold for some $y$, $R$ is in fact unrealizable as the element $\cons(1,y)$ cannot be given a valuation that is consistent with the definition of $R$. This does not happen with provably acyclic definitions because such functions can always be given a valuation consistent with their definitions on any model (see Theorem~\ref{thm:provably-acyclic}).

The completeness of \qialg\ arises from the fact that unfolding the (provably acyclic) definition of some $R$ ($R \in \Dd$) on $x$ amounts to a simulation of the ``computation'' of $R(x)$ in any model. Without provable acyclicity, we have shown that it is possible to construct well-defined definitions that are unrealizable in some models, rendering \qialg\ incomplete. 
}

\section{Related work}
\label{sec:relatedwork}

We discussed the relationship of our work to \liquid~\cite{Vazou18,liquid} and {\sc Leon}~\cite{SuterKoksalKuncak,BlancKuncakKneussSuter} extensively in Sections~\ref{sec:lh} and~\ref{sec:leon}. The mechanism for verifying specifications in the functional programming sub-language of Dafny~\cite{dafny,dafny-manual} is also similar. There is much prior work on techniques based on unfolding recursive definitions~\cite{amin2014computing,verimap,leino2013verifiedcalculations,acl2sedan}, going back to ACL2~\cite{acl2} and the NQTHM prover~\cite{boyer88}. The work on Set-Of-Support resolution~\cite{SetOfSupport1965,SetOfSupportImproved2021} is also similar, but it does not consider background theories.

The work in~\cite{suter10} shows that {\sc Leon}-like reasoning (and \qialg\ in this paper) is actually a decision procedure for certain restrictive logics. More precisely, it exhibits a logic over restricted classes of user-defined abstractions of ADTs to collections/measures in a decidable sort using catamorphisms, and shows that unfolding function definitions just \emph{once} followed by quantifier-free reasoning is a decision procedure. The classes of such abstractions (\emph{infinitely surjective and sufficiently surjective abstractions}) however are extremely semantically restrictive compared to \sys. In particular, as we show in Section~\ref{sec:boundarythms}, validity of \sys is undecidable, which argues this difference. The work in~\cite{Vazou18} shows that the PLE heuristic implemented in \slh is complete if an `equational proof' exists, but this result is much weaker than ours, and in fact PLE fails to prove simple theorems that \qialg~ can prove. 
\iffull
We furnished an example demonstrating this in Section~\ref{sec:lh-ple}.
\fi

The Why3 system~\cite{why3logic,filliatre13esop,why3casestudy2015} also verifies functional programs against contracts, but reduces verification conditions to first-order logic, integrating with several first order logic reasoning engines like Vampire~\cite{vampire}. FO provers such as Vampire~\cite{Hajdu2021} and Zipperposition~\cite{Cruanes2017} support reasoning about ADTs, even providing some automation for inductive reasoning. While it is also possible in our setting to use FO theorem provers to prove formulas of the form $\vc$, the practical effectiveness of such a reduction has not been evaluated and requires addressing some challenges, especially background theory reasoning; see~\cite{Reger2017-Vampire-as-SMT-solver}.

SMT solvers~\cite{z3,cvc4,bjornerthesis,Reynolds-Blanchette2017} provide powerful automation for logic reasoning, especially for quantifier-free fragments of decidable combinations of theories~\cite{nelson-oppen1979,tinelliharandi,calcofcomp}. There is work that develops decidable fragments by building over SMT solvers~\cite{suter10,RADA2013} as well as specialized decision procedures~\cite{MannaLFCS2007,Zhang2006,HosseinInterpolation2017,KapurInterpolation2006}. Extensions of the ADT theory have also been studied~\cite{kovacs17,RybinaVoronkov2001}.

Prior work on combining theories include Nelson-Oppen decidable combinations of theories~\cite{nelson-oppen1979,nelson80,tinelliharandi} and following work~\cite{WiesPiskac2009CombiningTheories,BaaderGhilardi2005ConnectingTheories,Ghilardi2004CombinedSatisfiability,Fontaine2007CombinationsBSR,Krstic2007CombinedSatisfiability,Tinelli2005NonstablyInfinite} extending this result. Local theory extensions~\cite{Ihlemann2008LocalTheoryExtensions,Sofronie-Stokkermans2009LocalTheoryExtensions} have also been employed for constructing decidable logics. Our work can be seen as reasoning with a particular fragment of quantified first-order logic (\sys) over combined theories using particular procedures (especially SMT solvers) that work well in practice in certain domains.

Techniques based on quantifier instantiation have been popular in automatic reasoning of quantified logics, including works on quantifier instantiation for SMT solvers~\cite{Reynolds16,barrett11,z3}. There are many methods to guide instantiation, such as triggers/E-matching~\cite{detlefs05,rummer12,amin2014computing,Moskal2009Triggers}, MBQI~\cite{Ge-deMoura2009CompleteInstantiationSMT}, etc. In general, systematic quantifier instantiation in the style of \qialg~ is not applicable to SMT solvers as the set of terms blows up exponentially. 

The work reported in~\cite{loding18} that shows completeness of a heuristic in practice called \emph{natural proofs}~\cite{qiu13,pek14} is closest to our work,
They show that natural proofs can be viewed as
reasoning in FOL by instantiating quantifiers using terms over a foreground uninterpreted universe followed by quantifier-free reasoning. They prove that this technique is complete for a \emph{safe fragment} of first-order logic. There are, however, fundamental differences in our work. First, the foreground sorts in our setting are ADT sorts and not uninterpreted. Second, the safe fragment identified in~\cite{loding18} is very restrictive as it disallows uninterpreted functions to involve background sorts, which in our setting would mean programs cannot have input parameters of the background sort, like integers. Finally, the quantifier instantiation strategy studied in~\cite{loding18} is much more liberal than in our work (and what tools like {\sc Liquid Haskell} and {\sc Leon} do). For example, if $\overline{t}$ is a set of terms that occur in a theorem, the instantiation in~\cite{loding18} will always instantiate the definition of $f$ on $\overline{t}$, while we will do so only when $f(\overline{t})$ occurs in the theorem. Consequently, the proof of our main theorem is quite complex and fundamentally different from the proof of completeness in~\cite{loding18}. 

Triggers are heuristic ways to control quantifier instantiation in SMT solvers, and SMT solvers as well as tools such as Boogie~\cite{boogie2-rustan} and Dafny~\cite{dafny} provide mechanisms for specifying triggers, both automatically~\cite{LeinoPit-Claudel2016StableTriggers} and manually~\cite{dafny-manual}. However, triggers are not simple quantifier instantiations~\cite{LeinoPit-Claudel2016StableTriggers,Moskal2009Triggers}. Further, trigger-based quantifier instantiation can be unpredictable and flaky. \rev{To the best of our knowledge, there has not been a substantial body of work on studying completeness for these approaches. Prior work provides some understanding of the structure of E-matching and triggers for completeness when applied to local theory extensions~\cite{DecidingLocalTheoryExtensions2015}, and more recent work develops a formal model of E-matching to prove termination~\cite{InstantiationTerminationEMatching2024}. In particular, the formal model in the latter work may be a useful starting point to explore showing completeness as well.}

\rev{In this work we constructed rogue models manually, and did not explore a systematic way of synthesizing such models from failed proofs. There has been some work on surveying the range of rogue nonstandard models that arise from failed proofs in practice~\cite{BlanchetteLPAR2010,EladSeplogicCompleteness2026}, as well as some synthesis algorithms for generating such models from shape `templates'~\cite{EladetalPOPL24,EladSeplogicCompleteness2026}. However, to the best of our knowledge there has not been dedicated work that establishes a database or taxonomy of rogue models, especially in the context of many-sorted FOL with background theories.} 

Automating induction has been explored in prior work~\cite{clam,isaplanner,hipspec,passmore20,Hajdu2020,Cruanes2017} for various specialized fragments~\cite{Unno2017CHCInduction}. In many cases, user help in the form of lemmas is still needed, though there is work on synthesizing inductive lemmas automatically~\cite{reynolds15,Weikun19,fossil,Sivaraman2022-LemmaSynthsesisCoq,EladShohamPOPL2025}.

\section{Discussion}
\label{sec:discussion}


\added{
This work is part of a larger program to give foundations to the theory and practice of automated reasoning for expressive logics, especially those that involve quantification and recursion. Expressive logics of this kind abound in program verification with unbounded program configurations or unbounded specifications. For example, programs manipulating data structures in pointer-based heaps operate over unbounded linked lists and trees. The domain of functional programs manipulating ADTs considered in this work is another example, where terms in the initial algebra are unbounded in size. Other examples include distributed systems with unbounded message queues or unbounded processes, or even specifications that deal with unbounded abstract/ghost values such as sequences or sets. Logics that express properties of such unbounded structures typically use either quantification (e.g., ``Every index in the array stores a positive number'') or recursion (e.g., ``The set of all reachable locations in memory is untainted''). 

Automated validity checking for expressive logics is the next frontier in automated verification. We compare this moment to the development of modern SMT solving, which was the result of research into practicable theory and algorithms for decision procedures for combinations of decidable quantifier-free theories. However, progress has since tapered off, and new approaches are needed. Looking at decidable fragments of these expressive logics might appear to be a natural alternative, but this has not worked out either. There are but a few exceptions such as the decidable theory of maps with pointwise updates~\cite{pointwisearrays}, the Array Property Fragment~\cite{Bradley2006Decidable}, and some decidable fragments of Separation Logic~\cite{BerdineCalcagnoOHearn2005,PiskacWiesZufferey2013,twbcade13}.

\smallskip
\mypara{Challenges} There are several challenges towards the development of an elegant theory in our program. First, expressive logics are almost always incomplete, i.e., there cannot exist sound proof systems or validity checking algorithms that can prove every valid formula. This often means that there may not be any theoretical insights for building useful proof systems, and from an empirical standpoint there may not exist any canonical approaches for automating validity. Second, even if individual theories are complete, combining them can result in incompleteness. For example, consider integers with addition and uninterpreted functions with equality. Each domain separately admits a complete theory, e.g., Presburger Arithmetic~(\cite{presburgertranslation1,presburgertranslation2}) and the theory of uninterpreted functions with equality (EUF) respectively. However, their combination is incomplete. Said differently, the set of valid theorems in arithmetic or uninterpreted functions (considered separately) is precisely the set of valid formulas in the individual theories. However, the set of valid theorems that involve both integers with addition and uninterpreted functions cannot be captured precisely by any recursive set of axioms! In particular, simply combining the axioms of Presburger arithmetic and EUF will result in a theory in which some valid theorems will not hold. Therefore, even automation that is complete for the union of the two theories would necessarily fail to prove some valid theorems\footnote{The reason for this is also attributable to the kind of Rogue Nonstandard Models we discuss in this work.}. These difficulties then lead to the following question: what is a reasonable theoretical standard to which we can hold automation for expressive logics?

From the practical side, a third challenge is to find a modular reasoning approach that can delegate reasoning over each theory to specialized solvers, similar to Nelson-Oppen combinations~\cite{nelson-oppen1979,nelson80} for decidable theories. Finding a set of reasonable theoretical conditions that would guarantee successful combination--- meaning that algorithms that are `good' for the individual theories can be combined into a `good' one for the combination--- appears hard.

\smallskip
\mypara{From practice to theory} This work goes from practice to theory. We present a theoretical foundation for a practical heuristic for reasoning about properties of functional programs and show that this technique is in fact complete procedure for reasoning about the combination of the underlying first-order theories, building a basis for why the technique works predictably well in practice. To achieve this end we developed the \sys fragment and the idea of provably acyclic definitions.

Our study is nontrivial since it focuses specifically on \qialg, which is a \emph{thrifty} unfolding strategy. Observe that even for the $\defns \rightarrow \varphi$ fragment that we study, one can achieve complete reasoning by instantiating $\defns$ on \emph{all} possible terms: all possible ADTs formed from all possible terms over the various sorts, all possible integers for the integer sort, etc. This naive strategy is in fact similar to the one implemented by Gilmore in the 1960s~\cite{Gilmore1960}, and it suffered significantly from the explosion of instantiations.

In contrast, instantiation performed by tools such as {\lh} and {\leon} is \emph{thrifty}, only instantiating function definitions on terms where their applications occur in the current formula. For instance, if $\varphi$ has only one occurrence of a function $f$, say applied to the tuple of terms $(t_1, \ldots, t_k)$, then the definition of $f$ will be instantiated \emph{only} on this tuple. Other terms of the sort, even constants, will not get instantiated. Even permutations of this tuple will not get instantiated. The thrifty strategy only instantiates on terms that naturally occur during \emph{computation} of function terms occurring in the formula.

Our completeness theorem proves that even this thrifty instantiation is  \emph{complete} for the \sys\ fragment. This means that it is as powerful as instantiating on all possible terms. Our result has practical consequences: thrifty instantiation is computationally less expensive, and tool designers can comfortably employ it without worrying about missing proofs. From a user's perspective, one can discourage users from providing instantiations that are not in the horizon of the thrifty strategy, i.e., the computational closure.

This paper is the second work in the program that shows that thrifty instantiations employed in practice turn out to be complete for first-order reasoning of programs. Earlier work on the foundations of \emph{Natural Proofs}~\cite{qiu13,pek14} for verifying heap-manipulating programs~\cite{loding18} showed that natural proofs, which is a thrifty instantiation strategy for reasoning about first-order logic with recursively defined functions, is complete when interpreting recursively defined functions under fixpoint semantics (which is first-order) as opposed to the intended least fixpoint semantics. The logics in the two works as well as the completeness theorems differ in many ways; we discuss these at the end of this section. However, considered broadly, these results suggest that first-order completeness may be the theoretical artifact that can guide the design of predictable automation for expressive logics.

\smallskip
\mypara{Going back to practice} There are several ripe directions to pursue for transferring insights from theory to practice. The first one is the idea that FO-completeness is a theoretical standard we can demand of automation even for incomplete logics. This is quite elegant, since FO-completeness means that validity procedures \emph{do not miss any proofs}. 
It shows that the procedures are not merely heuristics, and there won't be any silly failures that are product of ad-hoc algorithmic design.

A second realization is that FO-completeness can serve as a useful design principle for designing expressive logics with efficient practical automation. We know anecdotally that in both the world of functional programs as well as heap-manipulating programs, the corresponding FO-complete procedures are practically efficient, requiring only a few instantiations with terms of small height ($\sim$2-3). Recent work~\cite{flvOOPSLA25} investigates the validity of this principle for the design of new logics by designing two new heap logics (including a new separation logic) that admit FO-complete reasoning procedures. The work also evaluates these procedures for verifying a standard suite of data structure manipulating programs with rich specifications of their functional behavior and shows that the procedures are efficient.

Adopting FO-completeness as a design principle is not a trivial decision. If one is not prepared to prove complex technical results anew with each new proposed logic, one must essentially design logics that obey the restrictions placed by various fragments for which FO-completeness results are known. In the case of the work discussed above~\cite{flvOOPSLA25}, the authors make careful choices in order to define heap logics that admit FO-completeness. For example, in Separation Logic it is common to define a predicate $\mathit{list\_and\_length(x,n)}$ that holds when $x$ points to a linked list of length $n$. However, this does not fall into the FO-complete fragment leveraged by the authors~\cite{loding18}, so they must instead look for alternate formulations to express their specifications. Of course, the work ultimately shows that despite these technical restrictions it is indeed possible to provide rich functional specifications for heap-manipulating programs.

Similarly, if one were to design an FO-complete logic for verifying functional programs using the fragment developed in this work, the fact that \sys only allows \emph{definitions} is absolutely crucial. Though $\defns$ are universally quantified formulas, definitions do not constrain the space of models of $\defns \rightarrow \varphi$ in significant ways (on the standard model, they are well-defined, and on nonstandard models, they are always satisfiable; see Theorem~\ref{thm:provably-acyclic}). Relaxing this fragment to arbitrary universally quantified premises destroys completeness (follows from Theorem~\ref{thm:boundary-incompleteness} in Section~\ref{sec:boundarythms}). In particular, \sys does not allow the inclusion of quantified lemmas or quantified contracts in the antecedent of the implication. The landscape of choices is complex and needs careful scrutiny. In fact, although we presented our \sys fragment and the completeness of UQFR for it upfront, we spent more than a year on identifying this fragment!

\smallskip
Finally, a third direction for impacting practice is the possibility of using rogue nonstandard models to witness proof failure and to guide (semi-)automatic synthesis algorithms for inductive lemmas that eliminate the rogue models. Several recent works on inductive lemma synthesis leverage this insight differently; some utilize the first-order models returned by SMT solvers as approximate counterexamples and pose inductive lemma synthesis as a learning problem~\cite{fossil}, while others design a sound but incomplete procedure for synthesizing rogue nonstandard models~\cite{EladetalPOPL24} which they then use to inform the instantiation of specific induction principles~\cite{EladShohamPOPL2025}. We discuss some proposals for future work in this direction in Section~\ref{sec:conclusions}.

\medskip
\mypara{On the completeness theorems} We end this section with a few thoughts on the two FO-Completeness results~\cite{loding18,oopsla2023completeness}. The similarities between them are amply clear. The results also satisfy many of the desiderata we listed when discussing challenges above. In particular, they support black-box combinations with an SMT solver. Unlike other historical completeness results known in literature (see Section~\ref{sec:relatedwork}), the thrifty instantiation algorithms only deal with recursively defined functions and delegate the remaining quantifier-free reasoning to an SMT solver. They are completely agnostic to any particular axiomatization of the background theories and have no access to the specific decision procedures used. This is also interesting from a technical standpoint because it is an instance of a successful combination of quantified theories!

On the other hand, the two results also have many differences. In natural proofs for imperative programs~\cite{loding18}, we have an \emph{uninterpreted foreground sort} (for modeling arbitrary pointer-based heaps), have universally quantified formulas that quantify only over the foreground sort, and certain restrictions on the logic. In particular, completeness of formula-based quantifier instantiation is guaranteed when functions that involve a background sort in their domain do \emph{not} map to the foreground sort. The work presented in this paper assumes that the foreground sort is an ADT sort (FO-axiomatized), allows quantification only for defining functions (which must be provably acyclic), but allows universal quantification in these definitions to span over both the foreground ADT sort as well as the background sorts (i.e., parameters to functions can be of the background sort).

One open problem is whether there exists a more general result that extends both these results. For instance, it would be useful to have a result that allows for (user-written) universally quantified lemmas to be incorporated in a completeness result that defines a thrifty instantiation scheme for such lemmas (current tools like \lh ask for users to provide the lemmas as well as instantiations of them, which of course fits into the fragment defined in this paper). The work on natural proofs for imperative programs~\cite{loding18}, however, allows already for incorporation of lemmas, as long as quantification is only over the foreground sort, in its completeness result. On the other hand, the results in this paper allow for definitions to allow quantification over background sorts which is disallowed in the completeness results for natural proofs of imperative programs~\cite{loding18}. 
}

\rev{%
One of the most promising directions for combining the two results is the exploration of definitions involving quantification over background sorts, to yield a completeness result for imperative programs. There has been some recent progress in this direction. The work in~\cite{EladSeplogicCompleteness2026} shows a  completeness result for fold/unfold reasoning for a fragment of separation logic with inductive definitions along with background theories. The fragment is described as the 'Effectively Determined Heap' fragment, and is aesthetically similar to (but formally incomparable with) the idea of precise predicates introduced in early work on separation logic~\cite{ohearn04}. Formulas in this fragment essentially have determined heaps, and the work shows that these determined heaps--- and in fact the entire fragment--- admit a faithful translation of the fragment to first-order logic. Completeness then follows via the application of Herbrand's theorem. Notably, this work does not show completeness for the \emph{thrifty} instantiation schemes typically used in practice. In contrast, the salient feature of our work is the completeness result for thrifty instantiation that is much more efficient than the one yielded by Herbrand's theorem. We leave the exploration of the marriage of these ideas to show completeness of thrifty fold/unfold reasoning for Separation Logic and/or the EDH fragment to future work.
}

\section{Conclusions and future work}
\label{sec:conclusions}

In this work we investigate the theory behind the unreasonable effectiveness of a popular heuristic for reasoning with verification conditions generated during the verification of functional programs that compute over Algebraic Datatypes. We show that the heuristic is not just a set of intuitive design decisions, but is in fact complete for a first-order abstraction of the intuitive semantics a user has in their mind while utilizing the heuristic to prove theorems. We use our theoretical framework to explain the success of the heuristic, its failures, and formally explain the role of the user's intervention in helping the heuristic overcome its failures.

We believe our completeness result not only gives a theoretical foundation for heuristics used by practical verification tools, but also suggests a fundamentally new design paradigm for 
verification languages. The design of programming languages with specification languages guaranteed to be verifiable using complete techniques can lead to practical automation. 
Enriching our logic to datatypes beyond ADTs (e.g., abstract data types such as sets, maps, and queues) while supporting complete verification is also an interesting future direction.

Apart from the future directions mentioned above, a particularly interesting extension concerns Higher-Order Functions (HOFs). \slh and {\sc Stainless} do support defining HOFs, but such definitions are beyond the scope of the theory developed in this work. Tools like LH reason with HOFs by defunctionalization~\cite{Reynolds1972Defunctionalization}, i.e., converting them to FOL definitions by modeling function symbols as constants in a new sort and introducing an uninterpreted function $\mathit{apply}(f,\mathit{args})$ to model the application of a (higher-order) function $f$ on arguments $\mathit{args}$. However, simply defunctionalizing higher-order definitions and applying \qialg~ does not yield a completeness result for the appropriate higher-order logic. We conjecture that an analogous completeness theorem does in fact exist for ADTs and background theories with higher-order functions. 

That fact that rogue nonstandard models exist when inductive lemmas are needed (since our procedure is FO complete) provides an interesting direction to guide both users and tools towards new lemmas. In particular, one may be able to synthesize finite descriptions of rogue nonstandard models (similar to Section~\ref{sec:intuitive-ns-model}) using program synthesis, template-based synthesis using DSLs~\cite{EladetalPOPL24}, or even finite model finders~\cite{BlanchetteLPAR2010}. Such models can be presented to users as evidence of proof failure and lemmas suggested by users can be checked against them to evaluate whether they are useful. 

Exploiting these models to search automatically for inductive lemmas is also an interesting direction, especially in light of recent work that uses FO models to guide inductive lemma synthesis~\cite{fossil} for verifying heap manipulating programs. The Type-1 and Type-3 counterexamples in the work witness the falsehood of the goal and the non-inductiveness of candidate lemmas respectively. Rogue nonstandard models correspond precisely to Type-1 models (since they falsify the goal), and correspond to a variant of Type-3 models (helpful lemmas must not be inductive on rogue nonstandard models). We therefore believe that similar counterexample-guided synthesis techniques using rogue nonstandard models can lead to effective lemma discovery for verifying functional programs.


\bibliographystyle{ACM-Reference-Format}
\bibliography{refs}

\label{lastpage01}


\end{document}